\documentclass[journal,twoside,web]{ieeecolor}
\usepackage{tmi}
\usepackage{cite}
\usepackage{amsmath, amssymb,amsfonts}
\usepackage{algorithmic}
\usepackage{graphicx}
\usepackage{xcolor}
\usepackage{textcomp}
\usepackage{textgreek}
\usepackage{multirow}
\usepackage{booktabs}
\usepackage{bm}
\usepackage[colorlinks=true, 
allcolors=blue,
bookmarksopen=true,
bookmarksnumbered=true]{hyperref}
\usepackage{physics}
\usepackage{listings}
\usepackage{algorithm}

\usepackage{amsthm}
\usepackage{colortbl}
\usepackage{array}
\usepackage{makecell}
\usepackage{fancybox}
\usepackage{indentfirst}
\usepackage{xfp}

\usepackage{float}  %
\usepackage{subfigure}  %
\usepackage{color}

\newcommand\textred[1]{\textcolor{red}{\textbf{#1}}}
\newcommand\textblue[1]{\textcolor{blue}{\textbf{#1}}}

\newtheorem{theorem}{\bf Theorem}
\newtheorem{definition}{\bf Definition}
\newtheorem{corollary}{\bf Corollary}
\newtheorem{proposition}[theorem]{Proposition}

\def\BibTeX{{\rm B\kern-.05em{\sc i\kern-.025em b}\kern-.08em
    T\kern-.1667em\lower.7ex\hbox{E}\kern-.125emX}}
\begin{document}
\title{Trustworthy Multi-phase Liver Tumor Segmentation via Evidence-based Uncertainty}

\author{Chuanfei Hu, Tianyi Xia, Ying Cui, Quchen Zou, Yuancheng Wang, Wenbo Xiao, Shenghong Ju, Xinde Li, \IEEEmembership{Senior Member, IEEE}
\thanks{
%	This work was supported in part by the National Natural Science Foundation of China under Grant 62233003 and 62073072,
%	in part by the Key Projects of Key R\&D	Program of Jiangsu Province under Grant BE2020006 and Grant BE2020006-1,
%	in part by Shenzhen Natural Science	Foundation under Grant JCYJ20210324132202005,
%	in part by the National Natural Science Foundation of China under Grant 81830053,
%	in part by the National Key R\&D Program of China under Grant 2021YFF0501504.
	\textit{(Chuanfei Hu, Tianyi Xia and Ying Cui contributed equally to this work.)
	(Corresponding author: Xinde Li and Shenghong Ju.)}
}
\thanks{
	Chuanfei Hu, Quchen Zou, and Xinde Li are with the Key Laboratory of Measurement and Control of CSE, 
	School of Automation, Southeast University, Nanjing 210096, China (e-mail: cfhu@seu.edu.cn; 220201775@seu.edu.cn; xindeli@seu.edu.cn).
}
\thanks{
	Wenbo Xiao is with the Department of Radiology, the First Affiliated Hospital, Zhejiang University School of Medicine, Hangzhou 310058, China (e-mail: xiaowenbo@zju.edu.cn).
}
\thanks{
	Tianyi Xia, Ying Cui, Yuancheng Wang, and Shenghong Ju are with the Jiangsu Key Laboratory of Molecular and Functional Imaging, 
	Department of Radiology, Zhongda Hospital, School of Medicine, Southeast University, Nanjing 210009, China (e-mail: 504255977@qq.com; cuiy\_seu@163.com; yuancheng\_wang@163.com; jsh@seu.edu.cn).
}
}

\maketitle

\begin{abstract}
Multi-phase liver contrast-enhanced computed tomography (CECT) images convey the complementary multi-phase information for liver tumor segmentation (LiTS), which are crucial to assist the diagnosis of liver cancer clinically. 
However, the performances of existing multi-phase liver tumor segmentation (MPLiTS)-based methods 
suffer from redundancy and weak interpretability, % of the fused result,
resulting in the implicit unreliability of clinical applications.
In this paper, we propose a novel trustworthy multi-phase liver tumor segmentation (TMPLiTS), 
which is a unified framework jointly conducting segmentation and uncertainty estimation. 
The trustworthy results could assist the clinicians to make a reliable diagnosis.
Specifically, Dempster-Shafer Evidence Theory (DST) is introduced to parameterize the segmentation and uncertainty as evidence following Dirichlet distribution.
The reliability of segmentation results among multi-phase CECT images is quantified explicitly. 
Meanwhile, a multi-expert mixture scheme (MEMS) is proposed to fuse the multi-phase evidences, 
which can guarantee the effect of fusion procedure based on theoretical analysis. 
Experimental results demonstrate the superiority of TMPLiTS compared with the state-of-the-art methods. 
Meanwhile, the robustness of TMPLiTS is verified, 
where the reliable performance can be guaranteed against the perturbations. % in three scenarios.
\end{abstract}
\vspace{-2mm}


\vspace{-3mm}
\section{Introduction}
\label{sec:introduction}
\IEEEPARstart{L}{iver} cancer is one of the prevalent causes of cancer-induced death, 
resulting in a significant threat to human health. % []. 
Computed tomography (CT) imaging is the major technique for liver-related imaging tests, 
while accurate liver tumor segmentation from CT volumes can benefit liver therapeutic schedule planning and conduct more reliable clinical applications, 
such as prognostic metric for hepatic surgical procedures~\cite{nakayama2006automated, entezari2022promoting}, 
determination of radiation dose in liver tumor radioembolization~\cite{spina2019expected} and survival prediction~\cite{assouline2023volumetric}.

Clinically, manual delineation via experts is still an inevitable processing for liver tumor segmentation, 
which is tedious and labor-intensive. The subject delineation might be vulnerable with the increase of workload, resulting in unexpected missed and false detection. 
Thus, various computer-aided methods based on image processing and machine learning have been proposed for liver tumor analysis, 
such as shape parameterization \cite{Linguraru2012Tumor}, level set model \cite{Li2013Likelihood}, and support vector machine \cite{Xian2010identification}.
Nevertheless, the diversity of liver tumor in terms of appearance and location leads to difficulties for these conventional computer-aided methods.
How to construct a computer-aided automatic liver tumor segmentation is still an open issue.

\begin{figure}[!t]
	\centering
	\includegraphics[width=0.9\columnwidth]{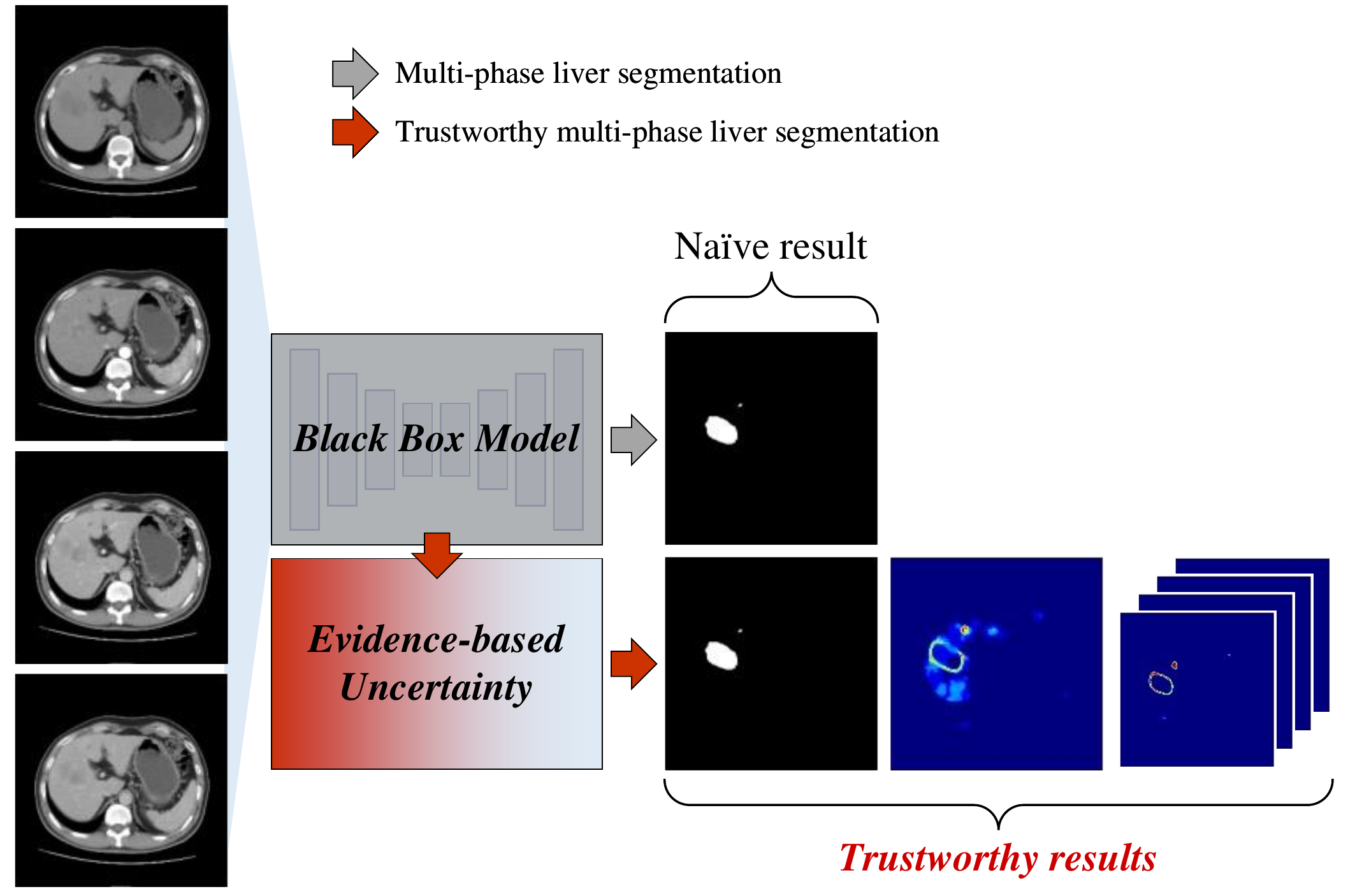}
	\caption{
		The insight of trustworthy multi-phase liver tumor segmentation (TMPLiTS). 
		Black box model might only provide the final segmentation result, 
		such na\"ive result could not interpret the reliability of model explicitly for multi-phase liver segmentation,
		resulting in the clinicians would make a diagnosis with slight hesitation.
		Inspired by~\cite{begoli2019need}, we argue that uncertainty estimation could reveal the reliability of black box model.
		Thus, a novel TMPLiTS is proposed, jointly achieving the segmentation and uncertainty estimation.
		Evidence-based uncertainty is introduced to quantify the uncertainty of prediction comprehensively,
		while trustworthy results could assist the clinicians to make a reliable diagnosis.
	} \label{fig:story}
	\vspace{-3mm}
\end{figure}

Recently, with the success of deep learning, many great efforts have be made to liver tumor segmentation, 
which can be categorized as single-phase-based~\cite{zhang2019light, lei2021defed, Di2022TD, lyu2022learning, zou2022graph} and multi-phase-based methods~\cite{vogt2020segmenting, zhang2021multi, zhang2021modality, zheng2022automatic}.
Single-phase-based methods construct deep learning model to segment liver tumors based on single-phase CT images.
However, the performances of these methods might not be satisfactory clinically due to the limited capability of single-phase CT imaging.
Compared with single-phase-based, 
multi-phase-based methods utilize the complementary information among different phases via contrast enhanced CT (CECT) imaging~\cite{chernyak2018liver},
which can capture the precise appearances of liver tumors. 
Specifically, feature-level fusion is the popular strategy for multi-phase-based methods to exploit complementary information among cross-phase features.
The complicated fusion modules are designed to bridge complementary relationships in terms of channel-wise and phase-wise features.
%However, the performances of multi-phase-based methods suffer from two limitations. 
However, the performances of multi-phase liver tumor segmentation (MPLiTS) methods suffer from two limitations. 
\emph{First}, these methods empirically aggregate information in the different dimensions of features and ignore the theoretical analysis,
resulting in redundancy and low efficiency of fusion structures.
\emph{Second}, the reliability of liver tumor segmentation on multi-phase CECT images might not be explicitly described.
Since clinicians often not only focus on precise segmentation results, 
but also want to know the reliability of each phase for the final results,
which is in line with the current trend in the medical image analysis community to build trustworthy, explainable, and reliable artificial intelligence (AI) model \cite{kundu2021ai, chen2022explainable, Gu2021CA, Liao2022Learning, Mao2022ImageGCN, Major2023On}.
Therefore, a challenging issue remains:

\setlength{\fboxsep}{5pt}
\noindent\shadowbox{\begin{minipage}[t]{0.94\columnwidth} {\it ``Can we design an MPLiTS method that can 
			jointly achieve reliable liver tumor segmentation and uncertainty estimation among multi-phase CECT images?''} \end{minipage}}

In this paper, we propose a novel trustworthy multi-phase liver tumor segmentation (TMPLiTS) on CECT images,
jointly conducting the segmentation and uncertainty estimation in a unified framework, as shown in Fig.~\ref{fig:story}.
%The evidence-based uncertainty 
Specifically, evidence-based uncertainty~\cite{sensoy2018evidential} is first introduced to jointly cast segmentation and uncertainty as evidence via expert layers,
modeled based on Dempster-Shafer Evidence Theory (DST) and Dirichlet distribution parameterization.
The expert layers can explicitly describe the reliability of liver tumor segmentation results from all phases. 
Then, a multi-expert mixture scheme (MEMS) is proposed based on a pixel-wise DST-based combination rule to fuse the multi-phase evidences,
in which theoretical analysis of MEMS guarantees the effect of fusion procedure which is validated by empirical results sufficiently.
To summarize, the main contributions are as follows: 
\begin{enumerate}
	\item A novel trustworthy multi-phase liver tumor segmentation (TMPLiTS) is proposed based on evidence-based uncertainty,
	 describing the reliability of segmentation results explicitly.
	To the best of our knowledge, we are among the first to represent the trustworthiness for multi-phase liver tumor segmentation (MPLiTS) on contrast enhanced CT (CECT) images.
	\item Multi-expert mixture scheme (MEMS) is designed to fuse the complementary results from multi-phase CECT images, 
	whose theoretical analysis guarantees the availability of fusion procedure.
	\item Extensive experiments conducted on our clinical in-house and external validation dataset demonstrate that 
	our method outperforms the state-of-the-art methods.
	Meanwhile, the robustness of TMPLiTS is verified which effectively promotes the reliability of MPLiTS against the three perturbed scenarios.
\end{enumerate}

\vspace{-2mm}
\section{Related Works}

\subsection{Multi-Phase Liver Tumor Segmentation}

Recent efforts of MPLiTS can be concluded as input-level fusion~\cite{ouhmich2019liver}, decision-level fusion~\cite{sun2017automatic, raju2020co},  and feature-level fusion methods~\cite{vogt2020segmenting, zhang2021multi, zhang2021modality, zheng2022automatic}.
The input-level fusion casts the multi-phase CECT images as the single image with multiply channels, 
while the segmentation network is adopted to the corresponding channels to predict the liver tumor regions.
Ouhmich \textit{et al.}~\cite{ouhmich2019liver} propose a cascaded convolutional neural network based on the U-Net architecture, 
and an input-level fusion strategy is designed for the arterial and portal venous phases images.
The two phase images are concatenated as an input data with multi-dimension.
The insight of decision-level fusion is to generate the final result in the prediction stage based on the multi-phase results,
where some intuitive operations are introduced, such as average operation. 
Sun \textit{et al.}~\cite{sun2017automatic} propose a multi-channel fully convolutional network (MC-FCN) to segment the liver tumors,
where the arterial, portal venous, and delayed phases of CECT images are exploited to model the networks independently.
The results of each network are aggregated by a fusion layer, and then, the softmax classifier is conducted to predict the final liver tumor regions.
Raju \textit{et al.}~\cite{raju2020co} develop a multi-phase framework based on co-heterogeneous training. 
The 15 combinations of non-contrast, arterial, venous, and delayed phases are considered in the prediction stage, 
where the average operation is conducted to derive a consensus result.
The feature-level fusion is a dominant strategy to MPLiTS, 
exploiting the complementary information among the different phases in terms of channel-wise and phase-wise features.
%in the procedure of feature extraction.
%The 
Vogt \textit{et al.}~\cite{vogt2020segmenting} propose a two-encoder U-Net architecture to segment the liver lesions in the arterial and portal venous phases,
where the independent encoded features of two phases are aggregated by the shared decoder.
Zhang \textit{et al.}~\cite{zhang2021multi} propose a deep learning-based method to aggregate the multi-phase information among the hierarchical features,
while an inpainting module is designed to refine the uncertain results by the neighboring features.
Zhang \textit{et al.}~\cite{zhang2021modality} design a mutual learning and modality-aware module for the arterial and venous phases. 
The attention maps are regressed by the modality-aware module to guide the feature fusion of different phases. 
Meanwhile, the intra and joint losses are conducted as a mutual learning strategy to share the interactive knowledge.
Zhang \textit{et al.}~\cite{zheng2022automatic} cast the pre-contrast, arterial, portal venous, and delayed phases as a time sequence,
while a shallow U-Net and a convolutional long short-term memory (C-LSTM) are constructed for the liver tumor segmentation.

These MPLiTS methods achieve the acceptable performances, however, 
the complicated fusion modules are designed empirically to aggregate the multi-phase information.
The theoretical analysis is weak, resulting in redundancy and low efficiency of fusion structures.
Furthermore, the interpretation of the results might not be presented explicitly, 
which is a valuable insight for clinicians to know the reliability of the {black box} model against the uncertain scenarios.

%multi-phase information and refine uncertain region segmentation.

%The decision-level fusion methods are designed to combine the separate results among the different phase.
%
%such a fusion strategy was explored for liver lesion segmentation in AP, PVP, and delayed phase CT. They demonstrated
%that segmentations inferred by fused networks were more accurate than those
%obtained for each separate phase.

\begin{figure*}[!t]
	\centering
	\includegraphics[width=0.9\textwidth]{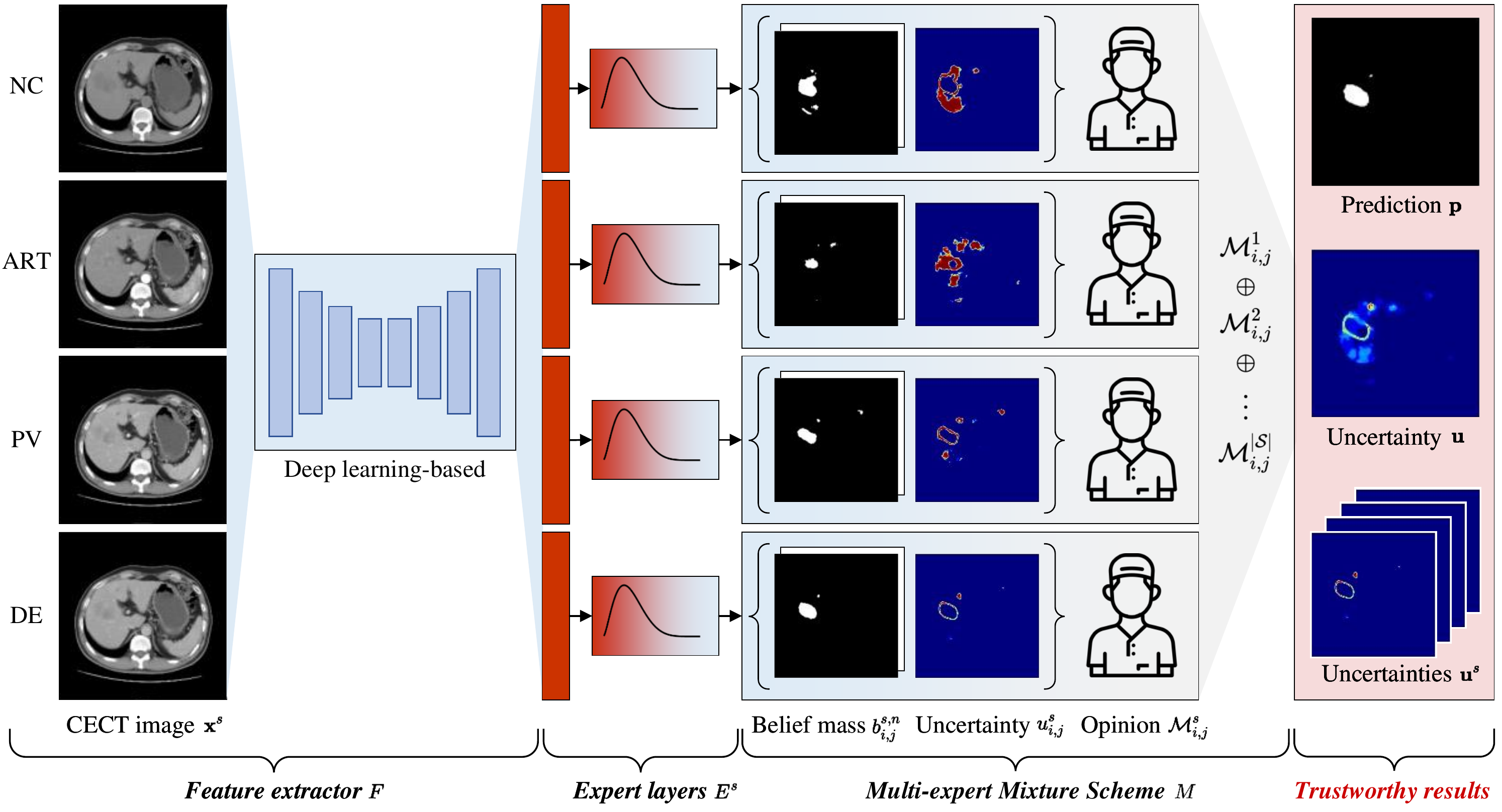}
	\caption{Overall framework of trustworthy multi-phase liver tumor segmentation, 
		jointly conducting the segmentation and uncertainty estimation in a unified framework. 
		The phases of CECT imaging involved in the framework are regarded as a set $\mathcal{S}$, 
		including non-contrast (NC), arterial (ART), portal venous (PV), and delayed (DE).  
		The feature map of $s$-th phase from CECT image is first represented via a deep learning-based feature extractor $F$.
		Then, the expert layers $E^{s}$ are modeled to introduce the evidence-based uncertainty based on DST and Dirichlet distribution parameterization, 
		which can explicitly cast the phase-wise reliability as the opinion $\mathcal{M}_{i,j}^{s}$, 
		composed of belief mass $b_{i, j}^{s,n}$ and uncertainty $u_{i,j}^{s}$ at the coordinate $(i, j)$.
		Meanwhile, MEMS is conducted to fuse the multi-phase opinions based on a pixel-wise DST-based combination rule.
		The trustworthy results can assist the clinicians to ``trust'' the model, conducting the reliable diagnosis among the multi-phase CECT images.
	} \label{fig:framework}
\vspace{-2mm}
\end{figure*}

\subsection{Uncertainty-Based Medical Image Segmentation}

The trustworthiness of {black box} model has attracted the considerable attention in the medical image analysis community~\cite{kundu2021ai, chen2022explainable, Gu2021CA, Liao2022Learning, Mao2022ImageGCN, Major2023On}, 
where uncertainty estimation is a popular insight to reveal the magnitude of trustworthiness~\cite{sensoy2018evidential, begoli2019need, Dong2021Evidential, Dong2022Multi}.
In medical image segmentation, the uncertainty-based methods are introduced to the paradigm of segmentation,
categorized as probabilistic-based~\cite{nair2020exploring, kohl2018probabilistic}, ensemble-based~\cite{lakshminarayanan2017simple}, and evidence-based~\cite{zou2022tbrats} methods.
Probabilistic-based methods construct a deep learning-based architecture to estimate the uncertainty based on a probability distribution, 
such as dropout, and conditional variational auto encoder (VAE).
The insight of ensemble-based methods is to train an ensemble of deep learning models to generate the uncertainty.
However, it might not be an available solution due to the high computation cost and low diversity.
Evidence-based methods conduct an evidential layer cascaded with the deep learning model, 
which can quantify the evidence-based uncertainty of segmentation results.
These uncertainty-based methods could be modeled for MPLiTS with an alternative strategy intuitively, 
as reported in~\cite{zou2022tbrats}.
The multi-phase CECT images are concatenated as a single image with multiple channels,
casting the multi-phase paradigm of MPLiTS as a single-phase liver tumor segmentation.

Compared with the intuitive strategy, there are two advantages of TMPLiTS over these uncertainty-based methods.
\textbf{\textit{First}}, the uncertainties of each phase are quantified independently, 
which can explicitly provide the multi-phase evidences to interpret the predictions,
assisting the reliable analysis of potential liver tumor on multi-phase images clinically.
There might be an intuitive solution that we could model an uncertainty-based method
%, such as TBraTS~\cite{zou2022tbrats}, 
for each phase image.
%The input of the model is an image of certain phase, while the corresponding uncertainty is estimated independently.
However, the complementary information among multi-phase results might not be considered adequately,
while the final uncertainty quantification would be obtained via intuitive operations, 
such as average operation, resulting in the unreliable uncertainty estimation.
Thus, the \textbf{\textit{second}} advantage of TMPLiTS is to extend the DST-based combination rule for the aggregation of the uncertainties among multi-phase CECT images, 
which can theoretically guarantee the reliability of fusion procedure.

%These uncertainty-based methods, reported in the medical segmentation task~\cite{zou2022tbrats},
%concatenate the multiple input data as a single data with multiple channels, 
%while the multi-phase paradigm of MPLiTS would be cast as a single-phase liver tumor segmentation.

%~\cite{nair2020exploring, lakshminarayanan2017simple, kohl2018probabilistic, zou2022tbrats}

\section{Methodology}

\subsection{Trustworthy Multi-Phase Liver Tumor Segmentation}
%\subsubsection{Framework}
%TMPLiTS aims to fuse the evidence-based uncertainty of each phase in a deep learning-based paradigm,
%ensuring the reliable performance of MPLiTS on CECT images.
The overall framework of TMPLiTS is shown in Fig.~\ref{fig:framework},
which consists of feature extraction, uncertainty estimation, and uncertainty fusion.
Specifically, given a set of CECT images $\mathcal{X} = \{ \mathbf{x}^{s} | {s \in \mathcal{S} \}}$ 
with a set of multi-phase $\mathcal{S} = \{\text{non-contrast}, \text{arterial}, \text{portal venous}, \text{delayed} \}$,
a CECT image $\mathbf{x}^{s} \in \mathbb{R}^{H \times W}$ of the $s$-th phase 
whose feature map $\mathbf{f}^{s} \in \mathbb{R}^{C \times H \times W}$ is first obtained via a deep learning-based feature extractor as follows:
\begin{equation}
	\vspace{-1mm}
	\mathbf{f}^{s} = F(\mathbf{x}^{s}),
\end{equation}
%where $p$ denotes the index of phase set $\mathcal{P} = \{\text{non-contrast}, \text{arterial}, \text{portal venous}, \text{delayed} \}$, 
where $C$ is the number of channels and $H \times W$ is the resolution of image.  % non-contrast, arterial, portal venous, and delayed
The parameters of $F(\cdot)$ are shared among all phases.
Then, the procedure of uncertainty estimation and fusion is conducted via expert layers $E^{s}(\cdot)$ and MEMS $M(\cdot)$, respectively.
Evidence-based uncertainty is introduced to model the deep evidence of liver tumor segmentation via expert layers,
in which the probabilities of segmentation are assumed to follow Dirichlet distribution.
Meanwhile, a pixel-wise DST-based combination rule is exploited to aggregate the uncertainties from all phase-wise experts.
Formally, the liver tumor prediction $\mathbf{p} \in \mathbb{R}^{N \times H \times W}$, uncertainty $\mathbf{u} \in \mathbb{R}^{H \times W}$, 
and uncertainties of all phases $\{ \mathbf{u}^{s} | {s \in \mathcal{S} \}}$ can be formulated as follows:
\vspace{-1mm}
\begin{equation}
	\vspace{-1mm}
	\left\{ \mathbf{p}, \mathbf{u}, \{ \mathbf{u}^{s} | {s \in \mathcal{S} \}} \right\} = M( \{ E^{s}(\mathbf{f}^{s}) | s \in \mathcal{S} \} ),
\end{equation}
where $N = | \mathcal{N} |$. 
$\mathcal{N}$ is a set of prediction categories $\mathcal{N} = \{  \text{background}, \text{hepatocellular carcinoma (HCC)} \}$.

\subsection{Evidence-based Uncertainty for MPLiTS}
\label{sec:eu}

Based on Dempster-Shafer Evidence Theory (DST), 
the evidence framework of Subjective Logic (SL) \cite{josang2016subjective} is formulated to explicitly associate 
belief and uncertainty with Dirichlet distribution parameterization.
SL is introduced as the basic of TMPLiTS to derive the belief and uncertainty of each phase, presenting the reliable and trustworthy results.

Specifically, we first denote a belief mass and uncertainty for the $s$-th phase prediction whose constraint is as follows:
\vspace{-1mm}
\begin{equation}
	\vspace{-1mm}
\sum_{n \in \mathcal{N}} b_{i, j}^{s, n}+u_{i, j}^{s}=1,
\end{equation}
where $b_{i, j}^{s, n} \in [0, \infty ) $ and $u_{i, j}^{s} \in [0, \infty )$ are the belief mass of the $n$-th category and the uncertainty, respectively.
$(i, j) \in (H, W)$ denotes the coordinate of $b_{i, j}^{s, n}$ and $u_{i, j}^{s}$ for the liver tumor prediction.
The belief mass is assigned via Dirichlet distribution with parameter $\boldsymbol{\alpha}_{i,j}^{s} = \{ \alpha^{s, n}_{i,j} | n \in \mathcal{N}\}$:
\vspace{-1mm}
\begin{equation}
	\vspace{-1mm}
	Dir(\mathbf{p}_{i,j}^s | \boldsymbol{\alpha}_{i,j}^{s})= \left\{
	\begin{aligned}
			&\frac{1}{B(\boldsymbol{\alpha}_{i,j}^{s})} \prod_{n \in \mathcal{N} } (p_{i,j}^{s,n})^{\alpha_{i,j}^{s, n}-1} & \!\!\!\!\!\! \text { for } \mathbf{p}_{i,j}^{s} \in \mathcal{P}^{s}_{i,j} \\ 
			&0  &\text { otherwise }
		\end{aligned}
	\right. ,
\end{equation}
where $B(\cdot)$ is the Beta function, and $\mathcal{P}^{s}_{i,j}  = \{ \mathbf{p}_{i,j}^{s} | \sum_{n \in \mathcal{N}} p_{i,j}^{s,n} = 1  \text{ and } p_{i,j}^{s,n} \in [0,1],  \forall n\}$ is the $N$-dimensional unit simplex. $p_{i,j}^{s,n}$ is a probability mass of $n$-th category prediction for the $s$-th phase.

Then, the evidence $e_{i,j}^{s,n}$ are linked with Dirichlet parameter $\boldsymbol{\alpha}_{i,j}^{s}$ by:
\begin{equation}
	\vspace{-1mm}
\alpha^{s, n}_{i,j} = e_{i,j}^{s,n} + 1, 
\end{equation}
where $e_{i,j}^{s,n} \in  [0, +\infty)$ is obtained directly from the expert layer with a non-negative activation function.
Thus, the belief mass and uncertainty can be formulated as:
\vspace{-1mm}
\begin{equation}
	\vspace{-1mm}
	b_{i, j}^{s, n} =  \frac{e_{i,j}^{s,n}}{ \sum_{n \in \mathcal{N}} \alpha_{i,j}^{s,n}}  \text{\quad and \quad} u_{i, j}^{s} = \frac{N}{\sum_{n \in \mathcal{N}} \alpha_{i,j}^{s,n}}.
\end{equation}
Following the Dirichlet assumption, the expectation of $p_{i,j}^{s,n}$ is given by:
\begin{equation}\label{equ:pro_dir}
	\mathbb{E}\left(p_{i,j}^{s,n}\right)=\frac{\alpha_{i,j}^{s,n}}{\sum_{n \in \mathcal{N}} \alpha_{i,j}^{s,n}}.
\end{equation}

\subsection{Expert Layers and Multi-expert Mixture Scheme}
%MEMS can be regarded as a modularization of evidence-based uncertainty estimation and fusion.  

\subsubsection{Uncertainty Estimation}
The expert layers are conducted to infer the evidences for each phase, 
which can be denoted as follows:
\vspace{-1mm}
\begin{equation}
\{ e_{i,j}^{s,n} | n \in \mathcal{N}\} = E^{s}(\mathbf{f}_{i,j}^{s}),
\vspace{-1mm}
\end{equation}
where $E^{s}(\cdot)$ denotes the $s$-th expert layer composed of two convolutional layers and Rectified Linear Units (ReLU) functions.
However, the convergence of model is not achieved in the preliminary experiments due to the exploding gradients.
We argue that the upper bound of evidences $e_{i,j}^{s,n}$ are not constrained, resulting in the large values.
Thus, the last ReLU function of $E^{s}(\cdot)$ is replaced by a composite function as follows:
\vspace{-1mm}
\begin{equation}
	f(x) = e^{\text{tanh}(x)},
\vspace{-1mm}
\end{equation}
where $\text{tanh}(\cdot)$ is hyperbolic tangent function.
Then, the $s$-th expert opinion $\mathcal{M}_{i,j}^{s} = \left\{\{ b_{i, j}^{s, n} | n \in \mathcal{N} \}, u_{i, j}^{s} \right\}$ can be gathered via Dirichlet distribution parameterization.
Meanwhile, we can obtain the uncertainty of $s$-th phase $\mathbf{u}^{s} = \{ u_{i, j}^{s} | (i, j) \in (H,W)\}$.

\subsubsection{Uncertainty Fusion} 
%Since the characteristic of multi-phase CECT imaging~\cite{chernyak2018liver}, 
%the opinions from multiply experts are based on the observations for different phases, 
Since the complementarity among expert opinions is inherent clinically~\cite{chernyak2018liver},
we extend the reduced Dempster's combination rule~\cite{Han2023Trusted} to MPLiTS in terms of pixel-wise, 
designing a mixture scheme for multi-expert opinions.
%The combination of joint opinion $\mathcal{M}_{i,j} = \{\{ b_{i, j}^{n} | n \in \mathcal{N} \}, u_{i, j} \}$ is fused by two opinions 
%$\mathcal{M}_{i,j}^{1} = \{\{ b_{i, j}^{1, n} | n \in \mathcal{N} \}, u_{i, j}^{1} \}$ and $\mathcal{M}_{i,j}^{2} = \{\{ b_{i, j}^{2, n} | n \in \mathcal{N} \}, u_{i, j}^{2} \}$
%with the following rule $\mathcal{M}_{i,j} = \mathcal{M}_{i,j}^{1} \oplus \mathcal{M}_{i,j}^{2}$. More specifically, the rule can be formulated as follows:
%\begin{equation}
%	b_{i, j}^{n} =(b_{i, j}^{1, n}b_{i, j}^{2, n} + b_{i, j}^{1, n} u_{i, j}^{2} + b_{i, j}^{2, n} u_{i, j}^{1})  / {(1 - C)}, u_{i, j} = { u_{i, j}^{1} u_{i, j}^{2}} / {(1 - C)},
%\end{equation}
%where $C = \sum_{n_{1} \neq n_{2}} b_{i, j}^{1, n_{1}}b_{i, j}^{2, n_{2}}$ is a coefficient to measure the conflict between $\mathcal{M}_{i,j}^{1}$ and $\mathcal{M}_{i,j}^{2}$.  
%$1/(1 - C)$ is a normalization factor.
\setcounter{definition}{0} % 将公式计数设置为0
\begin{definition}
	\textbf{Pixel-wise Reduced Dempster's Combination Rule	for $N$-Category Prediction at $(i, j)$.}
%	\begin{minipage}[t]{0.95\linewidth}
%	 \indent 
The combination of joint opinion $\mathcal{M}_{i,j} = \left\{\{ b_{i, j}^{n} | n \in \mathcal{N} \}, u_{i, j} \right\}$ is fused by two opinions 
$\mathcal{M}_{i,j}^{1} = \{\{ b_{i, j}^{1, n} | n \in \mathcal{N} \}, u_{i, j}^{1} \}$ and $\mathcal{M}_{i,j}^{2} = \{\{ b_{i, j}^{2, n} | n \in \mathcal{N} \}, u_{i, j}^{2} \}$
with the following rule:
\vspace{-2mm}
\begin{equation}
		\mathcal{M}_{i,j} = \mathcal{M}_{i,j}^{1} \oplus \mathcal{M}_{i,j}^{2}.
\vspace{-2mm}
\end{equation} 
%More specifically, 
Specifically, the combination rule of entities can be formulated as follows:
\vspace{-2mm}
\begin{equation}
		b_{i, j}^{n} = \frac{1}{1 - C}(b_{i, j}^{1, n}b_{i, j}^{2, n} + b_{i, j}^{1, n} u_{i, j}^{2} + b_{i, j}^{2, n} u_{i, j}^{1}),
\vspace{-2mm}
\end{equation}
\vspace{-2mm}
\begin{equation}
		u_{i, j} = \frac{1}{1 - C}{ u_{i, j}^{1} u_{i, j}^{2}},
\vspace{-2mm}
\end{equation}
where $C = \sum_{n_{1} \neq n_{2}} b_{i, j}^{1, n_{1}}b_{i, j}^{2, n_{2}}$ is a coefficient to measure the conflict between $\mathcal{M}_{i,j}^{1}$ and $\mathcal{M}_{i,j}^{2}$,  $\frac{1}{1 - C}$ is a normalization factor.
%	\end{minipage}

\end{definition}

According to the above-mentioned definition, the joint opinion $\mathcal{M}_{i,j}$ fused from multi-expert opinions can be formulated as follows:
\vspace{-3mm}
\begin{equation}
	\mathcal{M}_{i,j} = \mathcal{M}_{i,j}^{1} \oplus \mathcal{M}_{i,j}^{2} \oplus \cdots \mathcal{M}_{i,j}^{ \lvert \mathcal{S} \rvert}.
\vspace{-2mm}
\end{equation} 
Finally, the liver tumor prediction $\mathbf{p} = \{ {p}_{i,j}^{n} | (n,i,j) \in (N, H, W) \}$ 
and joint uncertainty $\mathbf{u}=\{u_{i, j} | (i,j) \in (H, W)\}$ can be obtained by $\mathcal{M}_{i,j}$ intuitively.

\subsubsection{Loss Function.} 
The loss function in the training procedure consists of phase-wise and mixture-wise losses,
which can be denoted as follows:
\begin{equation}
		\mathcal{L} = \overbrace{\frac{\lambda_{p}}{|\mathcal{S}|} \sum_{s \in \mathcal{S}}\mathcal{L}_{\gamma}(\mathbf{y}, \mathbf{p}^{s}, \boldsymbol{\alpha}^{s})}^{\text{phase-wise}} + \underbrace{\lambda_{m}\mathcal{L}_{\gamma}(\mathbf{y}, \mathbf{p}, \boldsymbol{\alpha})}_{\text{mixture-wise}}, 
\end{equation} 
where $\lambda_{p}$ and $\lambda_{m}$ are set to $0.5$ and $1$, respectively.
$\mathbf{p}^{s} = \{ {p}_{i,j}^{s, n} | (n,i,j) \in (N, H, W)\}$ is the liver tumor prediction of the $s$-th phase, $\mathbf{y} = \{ {y}_{i,j}^{n} | (n,i,j) \in (N, H, W)\}$ is the ground truth of liver tumor, 
and $\boldsymbol{\alpha} = \{ \alpha^{n}_{i,j} | (n,i,j) \in (N, H, W)\}$ is the Dirichlet distribution parameters. $\mathcal{L}_{\gamma}$ is composed of Dice loss $\mathcal{L}_{\text{D}}$ and Evidence loss $\mathcal{L}_{\text{E}}$, which can be denoted as follows:
\vspace{-1mm}
\begin{equation}
	\mathcal{L}_{\gamma}(\mathbf{y}, \mathbf{p}^{s}, \boldsymbol{\alpha}^{s}) = \mathcal{L}_{\text{D}}(\mathbf{y}, \mathbf{p}^{s}, \boldsymbol{\alpha}^{s}) + \mathcal{L}_{\text{E}}(\mathbf{y}, \mathbf{p}^{s}, \boldsymbol{\alpha}^{s}),
\end{equation}
where $\mathcal{L}_{\text{E}}$~\cite{sensoy2018evidential} is extended as follows: % for the image
\vspace{-1mm}
\begin{align}
	& \mathcal{L}_{\text{E}}(\mathbf{y}, \mathbf{p}^{s}, \boldsymbol{\alpha}^{s}) \nonumber \\%& = \frac{1}{HW}
	& =\sum_{i,j} \int\left[\sum_{n \in \mathcal{N}} -{y}_{i,j}^{n} \log \left({p}_{i,j}^{s, n}\right)\right] \frac{\prod_{n \in \mathcal{N}} (p_{i, j}^{s,n})^{\alpha_{i,j}^{s,n}-1} d \mathbf{p}_{i,j}^{s}}{B(\boldsymbol{\alpha}_{i,j}^{s})}  \nonumber \\
	& = \sum_{i,j} \sum_{n \in \mathcal{N}} {y}_{i,j}^{n} \left[\psi( \sum_{n \in \mathcal{N}} \alpha_{i,j}^{s,n} )-\psi(\alpha_{i,j}^{s,n}) \right],
%	& =  \frac{1}{HW} \sum_{i,j} \sum_{n \in \mathcal{N}} {y}_{i,j}^{n} \left[\psi( \sum_{n \in \mathcal{N}} \alpha_{i,j}^{s,n} )-\psi(\alpha_{i,j}^{s,n}) \right],
\end{align}
where $\psi(\cdot)$ is the digamma function, and $\sum_{i,j}$ is a brief symbol of $\sum_{(i,j) \in (H, W)}$.

\subsection{Theoretical Analysis}
The advantage of MEMS is to conduct the mixture scheme for multi-phase opinions based on theoretical guarantees,
where four propositions can be induced in terms of \textit{prediction accuracy} and \textit{uncertainty estimation}. 

Specifically, we simplify the notation of original opinion at coordinate $(i,j)$ for the MPLiTS as:
\begin{equation}
	\mathcal{M}^{o} = \left\{\{ b^{o,n} | n \in \mathcal{N} \}, u^{o} \right\}. 
\end{equation}
%$\mathcal{M}^{o} = \left\{\{ b^{o,n} | n \in \mathcal{N} \}, u^{o} \right\}$, 
The theoretical analysis is conducted to investigate whether the combination of another opinion
\begin{equation}
	\mathcal{M}^{a} = \left\{\{ b^{a,n} | n \in \mathcal{N} \}, u^{a} \right\}
\end{equation}
%$\mathcal{M}^{a} = \left\{\{ b^{a,n} | n \in \mathcal{N} \}, u^{a} \right\}$ 
would deteriorate the performance of prediction.
The new opinion $\mathcal{M}$ is denoted as followed:
\begin{equation}
	\mathcal{M} = \mathcal{M}^{o} \oplus \mathcal{M}^{a} = \left\{\{ b^{n} | n \in \mathcal{N} \}, u \right\}.
\end{equation}
%$\mathcal{M} = \mathcal{M}^{o} \oplus \mathcal{M}^{a} = \left\{\{ b^{n} | n \in \mathcal{N} \}, u \right\}$.
More specifically, $\mathcal{M}^{o}$ and $\mathcal{M}^{a}$ can be regarded as the expert opinions from multi-phase CECT images at coordinate $(i,j)$.
The four propositions~\cite{Han2023Trusted} are concluded as follows:
%\textit{Prediction accuracy.} 
(1) The combination of two expert opinions can potentially improve the prediction accuracy of the model.
% integrating an additional opinion into the original opinion can potentially improve the classification accuracy of the model.
(2) The possible degradation of performance is limited under mild conditions, when the original opinion is fused with an additional opinion.
% when integrating an additional opinion into the source opinion, the possible performance degradation of the model is limited under mild conditions.
%\textit{Uncertainty estimation.} 
(3) The uncertainty of new opinion would be reduced by fusing the another opinion, 
and (4) also be large when the uncertainties of two opinions are both large.

%\begin{itemize}
%	\item[$\bullet$] \textbf{Prediction accuracy.} 
%	(1) The combination of two expert opinions can potentially improve the prediction accuracy of the model.
%	% integrating an additional opinion into the original opinion can potentially improve the classification accuracy of the model.
%	(2) The possible degradation of performance is limited under mild conditions, when the original opinion is fused with an additional opinion.
%	% when integrating an additional opinion into the source opinion, the possible performance degradation of the model is limited under mild conditions.
%	\item[$\bullet$] \textbf{Uncertainty estimation.} 
%	(3) The uncertainty of new opinion would be reduced by fusing the another opinion, 
%	and (4) also be large when the uncertainties of two opinions are both large.
%\end{itemize}

\begin{proposition}
	Under the conditions $b^{a,t} \geq b^{o,m}$, 
	where $b^{a,t}$ is the belief mass of ground truth category $t$ and $b^{o,m}$ is the largest belief mass in $\{ b^{o,n} | n \in \mathcal{N} \}$,
	the new opinion satisfies $b^{t} \geq b^{o,t}$.
\end{proposition}
\begin{proof}
	\begin{align*}
		& b^{t} = \frac{b^{o,t}b^{a,t}+b^{o,t}u^{a} + b^{a,t}u^{o}}{\sum_{n \in \mathcal{N}} b^{o,n}b^{a,n} + u^{a} + u^{o} - u^{o}u^{a} }  \\  
		& \geq \frac{b^{o,t}(b^{a,t} + u^{a} + u^{o}) }{ b^{o,m}(1-u^{a}) + u^{a} + u^{o} - u^{o}u^{a} } \\
		& \geq \frac{b^{o,t}(b^{a,t} + u^{a} + u^{o}) }{ b^{o,m} + u^{a} + u^{o} } \geq b^{o,t}.   
	\end{align*} 
\end{proof}
\begin{corollary}
	Since a positive correlation between belief mass $b^{n}$ and prediction $p^{n}$ of $\mathcal{M}$ mentioned in Section~\ref{sec:eu}, 
	the prediction accuracy can be improved potentially.
\end{corollary}
\begin{proposition}
	$b^{o,t} - b^{t}$ has a negative correlation with $u^{a}$.
	When $u^{a}$ is large, the upper bound of $b^{o,t} - b^{t}$ is reduced, and the possible degradation would be limited.
\end{proposition}
\begin{proof}
	\begin{align*}
		& b^{o,t} - b^{t} = b^{o,t} - \frac{b^{o,t}b^{a,t}+b^{o,t}u^{a} + b^{a,t}u^{o}}{\sum_{n \in \mathcal{N}} b^{o,n}b^{a,n} + u^{a} + u^{o} - u^{o}u^{a} }  \\
		& \leq b^{o,t} - \frac{b^{o,t}u^{a}}{ b^{a,m} + u^{a} + u^{o} - u^{o}u^{a} } \\
		& \leq \frac{b^{o,t}u^{a}}{ 1+u^{o} - u^{o}u^{a}} = b^{o,t}\frac{1+u^{o}}{ \frac{1}{1-u^{a}} + u^{o}}.  
	\end{align*} 
\end{proof}
\begin{proposition}
	After the fusion of the another opinion $\mathcal{M}^{a}$, the uncertainty $u$ of new opinion $\mathcal{M}$ would be reduced.
\end{proposition}
\begin{proof}
	\begin{align*}
		& u = \frac{u^{o}u^{a}}{\sum_{n \in \mathcal{N}} b^{o,n}b^{a,n} + u^{a} + u^{o} - u^{o}u^{a}} \\
		& \leq \frac{u^{o}u^{a}}{u^{a}+u^{o}-u^{o}u^{a}} \leq u^{o}. 
	\end{align*} 
\end{proof}
\begin{proposition}
	The uncertainty $u$ of new opinion $\mathcal{M}$ has a positive correlation with $u^{a}$ and $u^{o}$.
\end{proposition}
\begin{proof}
	\begin{align*}
		& u = \frac{u^{o}u^{a}}{\sum_{n \in \mathcal{N}} (b^{a,n}b^{o,n}+b^{a,n}u^{o}+b^{o,n}u^{a})+u^{o}u^{a}} \\
		& = \frac{1}{\sum_{n \in \mathcal{N}}( \frac{b^{a,n}b^{o,n}}{u^{o}u^{a}} + \frac{b^{a,n}}{u^{a}} + \frac{b^{o,n}}{u^{o}} ) + 1}. 
	\end{align*}
\end{proof}

\begin{table*}[!t]
%	\definecolor{whitesmoke}{gray}{.95}
	\definecolor{whitesmoke}{RGB}{239,246,255}
	\centering
	\caption{Comparison of TMPLiTS with other state-of-the-art methods. 
		The average results and standard deviations are both reported for five-fold cross-validation.
		The best and the second best performances are marked with \textred{red} and \textblue{blue}, respectively.
		$\blacklozenge$ denotes a variant of TMPLiTS with independent feature extractors $F^{s}(\cdot)$. 
		% on $\mathcal{T}_{\text{\textit{In}}, \text{\textit{Va}}}$ and $\mathcal{T}_{\text{\textit{Ex}}, \text{\textit{Va}}}$.
	}\label{tab:cmp_sota}
	\begin{tabular}{ccccccc}
		\Xhline{1pt} & & & & & &\\[-3mm] 
		\multicolumn{2}{c}{\multirow{2}{*}{Methods}} & \multicolumn{2}{c}{$\mathcal{T}_{\text{\textit{In}}, \text{\textit{Va}}}$} & \multicolumn{2}{c}{$\mathcal{T}_{\text{\textit{Ex}}, \text{\textit{Va}}}$} & \multirow{2}{*}{Memory footprint}\\ %\cmidrule(){3-14} %\cmidrule(r){4-4} \cmidrule(lr){5-6} \cmidrule(lr){7-10} \cmidrule(lr){11-14}
		\multicolumn{2}{c}{}                         & DGS        & DCS       & DGS        & DCS &\\
		\hline %& \cellcolor{whitesmoke}& \cellcolor{whitesmoke}& \cellcolor{whitesmoke}& \cellcolor{whitesmoke}& \cellcolor{whitesmoke}& \cellcolor{whitesmoke}\\[-2.5mm]
		%		\rowcolor{whitesmoke}
		\multirow{3}{*}{MPLiTS-based}            & \cellcolor{whitesmoke}SA-URI~\cite{zhang2021multi}             & \cellcolor{whitesmoke}82.21(0.95)          & \cellcolor{whitesmoke}80.91(0.96)         & \cellcolor{whitesmoke}{77.76}(0.76)          & \cellcolor{whitesmoke}70.95(1.15)       &\cellcolor{whitesmoke}956.9 MB\\
		& MAML~\cite{zhang2021modality}            & 80.91(0.82)          & 78.96(1.41)         & 76.96(1.45)          & 69.91(1.43)         & 113.2 MB\\
		& \cellcolor{whitesmoke}CLSTM~\cite{zheng2022automatic}            & \cellcolor{whitesmoke}79.92(1.23)          & \cellcolor{whitesmoke}78.54(1.56)         & \cellcolor{whitesmoke}77.49(1.24)          & \cellcolor{whitesmoke}{70.70}(1.35)      & \cellcolor{whitesmoke}302.5 MB\\
		\hline %& & & & &  &\\[-2.5mm]
		\multirow{4}{*}{Uncertainty-based}  & UE~\cite{lakshminarayanan2017simple}              & \textblue{82.59}(0.43)          & \textred{81.71}(1.12)         & 78.18(1.78)          & 70.42(1.81)   & 255.9 MB\\
		& \cellcolor{whitesmoke}   PU~\cite{kohl2018probabilistic}        & \cellcolor{whitesmoke}78.87(1.15)          & \cellcolor{whitesmoke}77.02(1.21)         & \cellcolor{whitesmoke}76.91(1.44)          & \cellcolor{whitesmoke}67.75(2.04)      & \cellcolor{whitesmoke}104.4 MB\\
		&  DU~\cite{nair2020exploring}           &81.63(1.34)          &\textblue{81.64}(1.10)         & 76.68(1.47)          & 69.38(1.33)         & 29.7 MB\\
		& \cellcolor{whitesmoke}TBraTS~\cite{zou2022tbrats}            & \cellcolor{whitesmoke}81.09(0.97)          & \cellcolor{whitesmoke}80.25(1.49)         & \cellcolor{whitesmoke}\textblue{78.51}(1.12)          & \cellcolor{whitesmoke}\textred{72.00}(1.05)          & \cellcolor{whitesmoke}51.3 MB\\
		\hline %& & & & & &\\[-2.5mm]
		\multirow{2}{*}{Ours} & TMPLiTS             & 81.49(1.74)          & 79.74(1.24)         & 78.18(0.82)          & 71.72(1.38)   & 51.7 MB\\
		& \cellcolor{whitesmoke}TMPLiTS$\blacklozenge$            & \cellcolor{whitesmoke}\textred{82.60}(1.68)          & \cellcolor{whitesmoke}81.07(1.87)         & \cellcolor{whitesmoke}\textred{79.20}(0.86)          & \cellcolor{whitesmoke}\textblue{71.88}(1.28)   & \cellcolor{whitesmoke}205.3 MB\\
		%		\hline  
		\Xhline{1pt}
	\end{tabular}
\vspace{-3mm}
\end{table*}

\section{Experiments}

\subsection{Experimental Setup}

\subsubsection{Dataset} 
We evaluate TMPLiTS on both in-house dataset and external dataset. 
The in-house dataset collects 388 patients 
with 1552 multi-phase CECT volumes from The First Affiliated Hospital of Zhejiang University. 
All volumes are acquired by Philips iCT 256 scanners with non-contrast, arterial, portal venous, and delayed phases. 
The in-plane size of volumes is $512 \times 512$ with spacing ranges from 0.560 mm to 0.847 mm, 
and the number of slices ranges from 25 to 89 with spacing 3.0 mm.
The volumes of four phases are co-registered into venous phase by Elastix~\cite{Klein2010elastix} toolbox.
The external dataset consists of 82 patients with 328 multi-phase CECT volumes by Philips iCT 256 scanners from Zhongda Hospital Southeast University. 
The in-plane size of volumes is $512 \times 512$ with spacing ranges from 0.601 mm to 0.851 mm, 
and the number of slices ranges from 36 to 139 with spacing 2.5 mm.
The procedure of registration is the same as the in-house dataset. 
The ground truths of all volumes are annotated by two radiologists (with 10 years and 20 years of experiences in liver imaging, respectively),
where the HCC lesions are outlined in the delayed phase, with reference to the other phases.

\subsubsection{Evaluation Protocols} 
The evaluation tasks of TMPLiTS are conducted in terms of validity and reliability, 
which can be denoted as $\{ \mathcal{T}_{d,t} | d \in \{ \text{\textit{In}}, \text{\textit{Ex}} \}, t \in \{ \text{\textit{Va}}, \text{\textit{Re}} \}\}$.
\textit{In} and \textit{Ex} denote the dataset from the in-house and external, respectively.
\textit{Va} means the evaluation task of validity, and \textit{Re} denotes the evaluation task of reliability.
For instance, $\mathcal{T}_{\text{\textit{In}}, \text{\textit{Va}}}$ is the evaluation task of validity on the in-house dataset.
Specifically, the protocols of the evaluation tasks are as follows:
\begin{itemize}
	\item[$\bullet$] $\mathcal{T}_{\text{\textit{In}}, t}$, the multi-phase CECT volumes of in-house dataset are divided into five folds with the same ratio in terms of patient, 
	where five-fold cross-validation is employed to evaluate the proposed method. 
	\item[$\bullet$] $\mathcal{T}_{\text{\textit{Ex}}, t}$, the parameters of models trained on in-house dataset with five folds are frozen, 
	while the external dataset used as the validation dataset.
	\item[$\bullet$] $\mathcal{T}_{d, \text{\textit{Va}}}$, Dice global score (DGS) and Dice per case score (DCS) are utilized to evaluate the performance of liver tumor segmentation. 
	DCS is an average dice score in terms of patient case, while DGS is conducted across all dice scores of each slice.
	\item[$\bullet$] $\mathcal{T}_{d, \text{\textit{Re}}}$,  expected calibration error (ECE)~\cite{jungo2019assessing} and uncertainty-error overlap (UEO)~\cite{jungo2019assessing} are utilized to verify the reliability of trained model against the various perturbations. 
	Since the values of ECE might not be various obviously, $- \ln (x)$ is utilized to design a negative-logarithm ECE function.
\end{itemize}

\subsubsection{Implementation Details}
The experiments are conducted on a work station with NVIDIA Tesla A100 GPUs. 
The proposed method is implemented based on PyTorch deep learning framework.
The backbone of feature extractor $F(\cdot)$ is U-Net~\cite{ronneberger2015u},
where the last feature with 64 channels. 
In the training stage, the learnable parameters of TMPLiTS, including feature extractor and expert layers, 
are trained from scratch via Adam~\cite{kingma2015adam} with a weight decay of 1e-5. 
The total training epochs are 400, and the initial learning rate is 5e-4, 
while the learning rate is adjusted by cosine annealing cycles with the periodic iterations of 60 epochs.
The batch size of training is 4. 
The training samples are scaled as $224 \times 224$, while random rotation of -5 to 5 degrees is used as data argumentation.
The window width and window level are set to 140 and 40, respectively.
Such training strategy is utilized both for $\mathcal{T}_{\text{\textit{Ex}}, t}$ and $\mathcal{T}_{\text{\textit{In}}, t}$.

\subsection{Comparisons with Other Methods}

We compare TMPLiTS and its variant (TMPLiTS$\blacklozenge$) with 7 state-of-the-art methods in terms of validity and reliability, % state-of-the-art methods
which can be categorized as MPLiTS-based~\cite{zhang2021multi, zhang2021modality, zheng2022automatic} 
and uncertainty-based~\cite{nair2020exploring, lakshminarayanan2017simple, kohl2018probabilistic, zou2022tbrats} methods.
TMPLiTS$\blacklozenge$ is an alternative version of TMPLiTS, 
whose feature extractor $F(\cdot)$ is adapted as the independent feature extractors $F^{s}(\cdot)$ for each phase.

\subsubsection{Evaluation of Validity on $\mathcal{T}_{\text{\textit{In}}, \text{\textit{Va}}}$ and $\mathcal{T}_{\text{\textit{Ex}}, \text{\textit{Va}}}$} 
As reported in Table~\ref{tab:cmp_sota}, 
we observe that TMPLiTS achieves the competitive performances of MPLiTS in terms of DCS and DGS on both $\mathcal{T}_{\text{\textit{In}}, \text{\textit{Va}}}$ and $\mathcal{T}_{\text{\textit{Ex}}, \text{\textit{Va}}}$,
while the memory footprint of our method is tolerant compared with SA-URI, CLSTM, and UE.
Furthermore, the performances of uncertainty-based methods are superior to MPLiTS-based methods on $\mathcal{T}_{\text{\textit{Ex}}, \text{\textit{Va}}}$.
It means that there might be a slight advantage of uncertainty-based methods to adapt ``\textit{unseen}'' instances in real-world clinical applications.

%\subsubsection{Evaluation of reliability on $\mathcal{T}_{\text{\textit{In}}, \text{\textit{Va}}}$ and $\mathcal{T}_{\text{\textit{Ex}}, \text{\textit{Va}}}$}
\begin{figure*}[!t]
	\centering
	\subfigure[Gaussian noise $\mathcal{T}_{\text{\textit{In}}, \text{\textit{Re}}_\text{noise}}$ with $\sigma^{2}_{n} = \{ 0.03, 0.05, 0.1, 0.2 \}$. ]{
		\begin{minipage}[t]{0.48\columnwidth}
			\centering
			%			\hspace{-4.8mm}
			\includegraphics[width=1\linewidth]{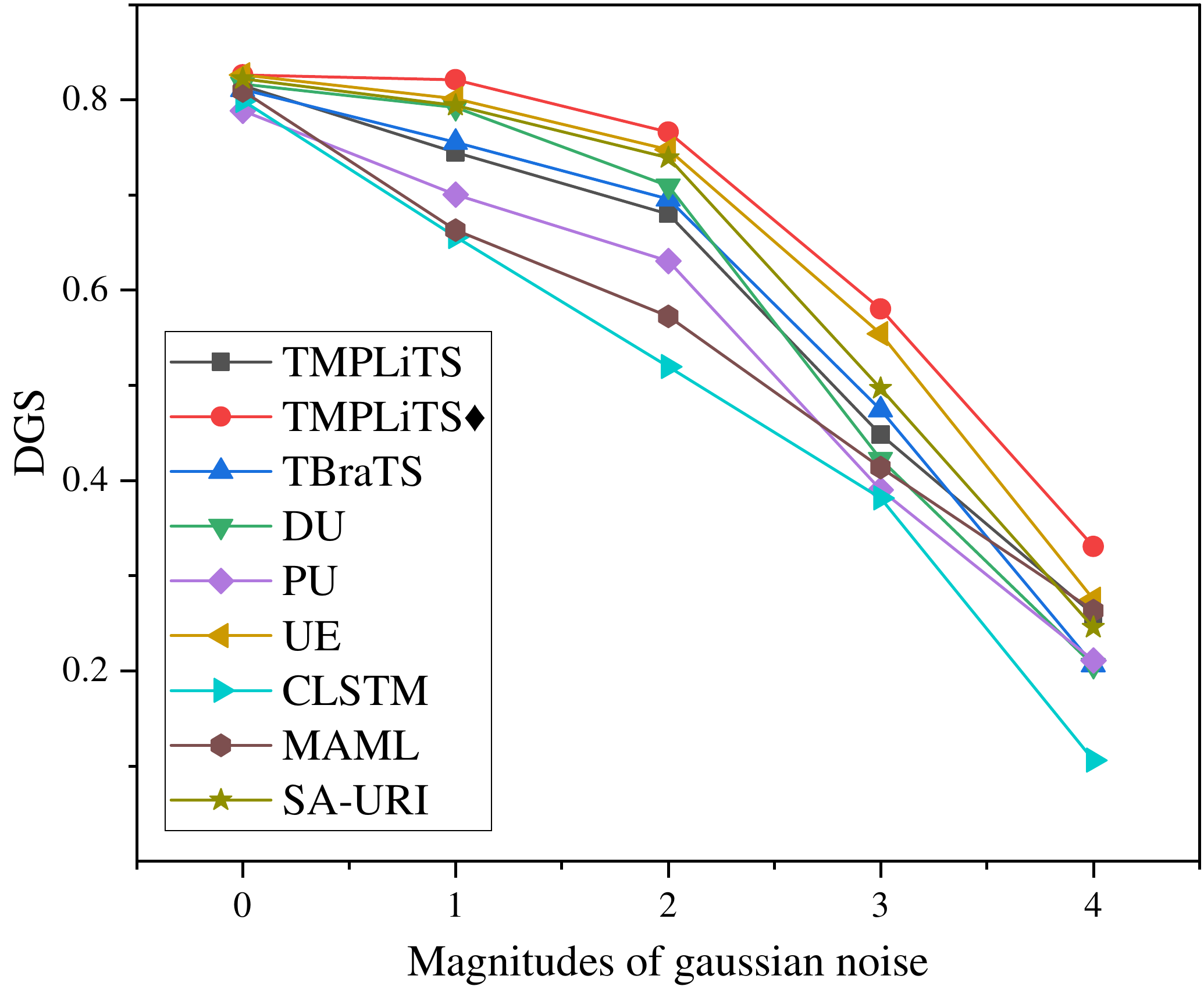}
		\end{minipage}%
		\begin{minipage}[t]{0.48\columnwidth}
			\centering
			%			\hspace{-4.8mm}
			\includegraphics[width=1\linewidth]{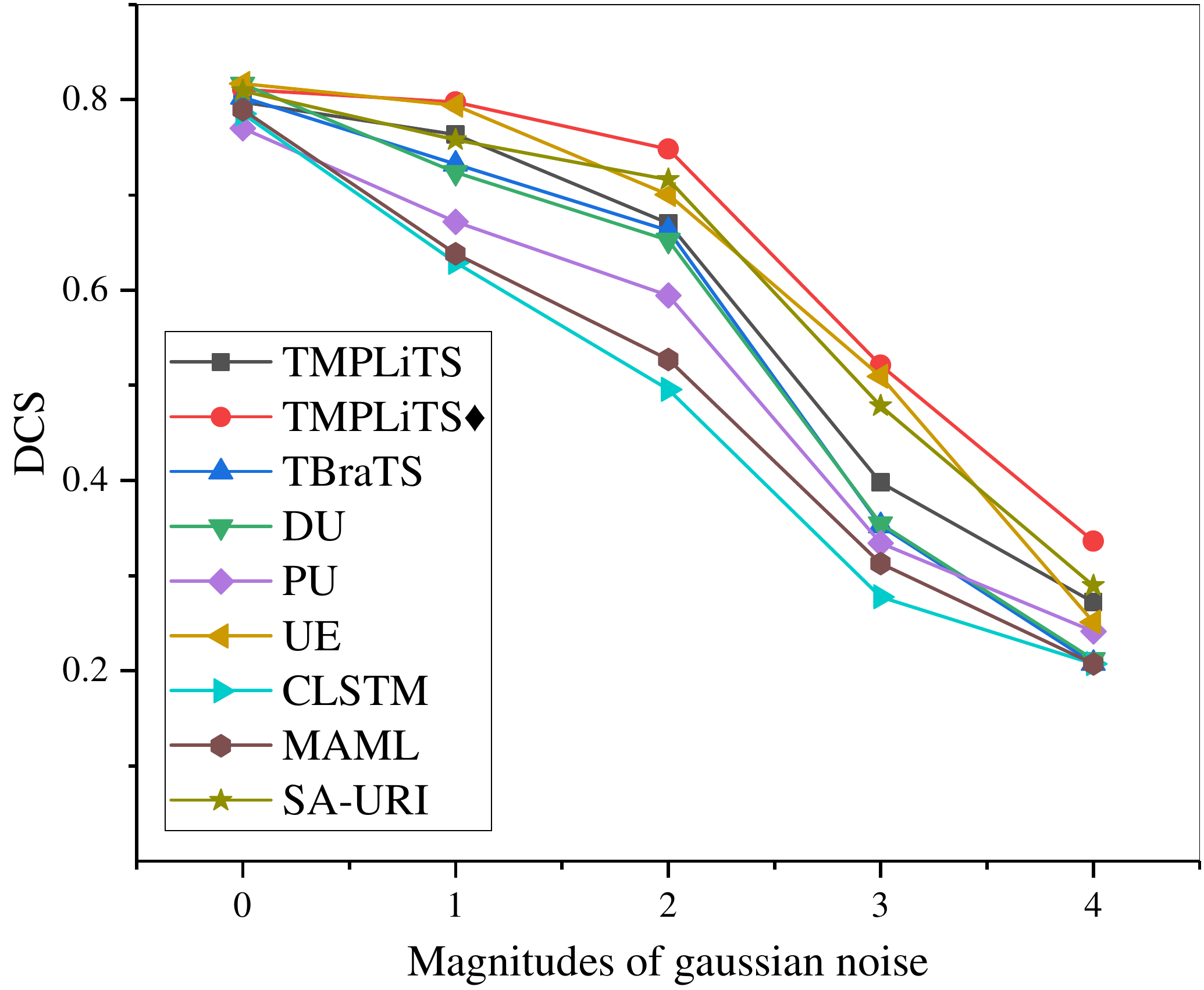}
		\end{minipage}%
		\begin{minipage}[t]{0.48\columnwidth}
			\centering
			%			\hspace{-4.8mm}
			\includegraphics[width=1\linewidth]{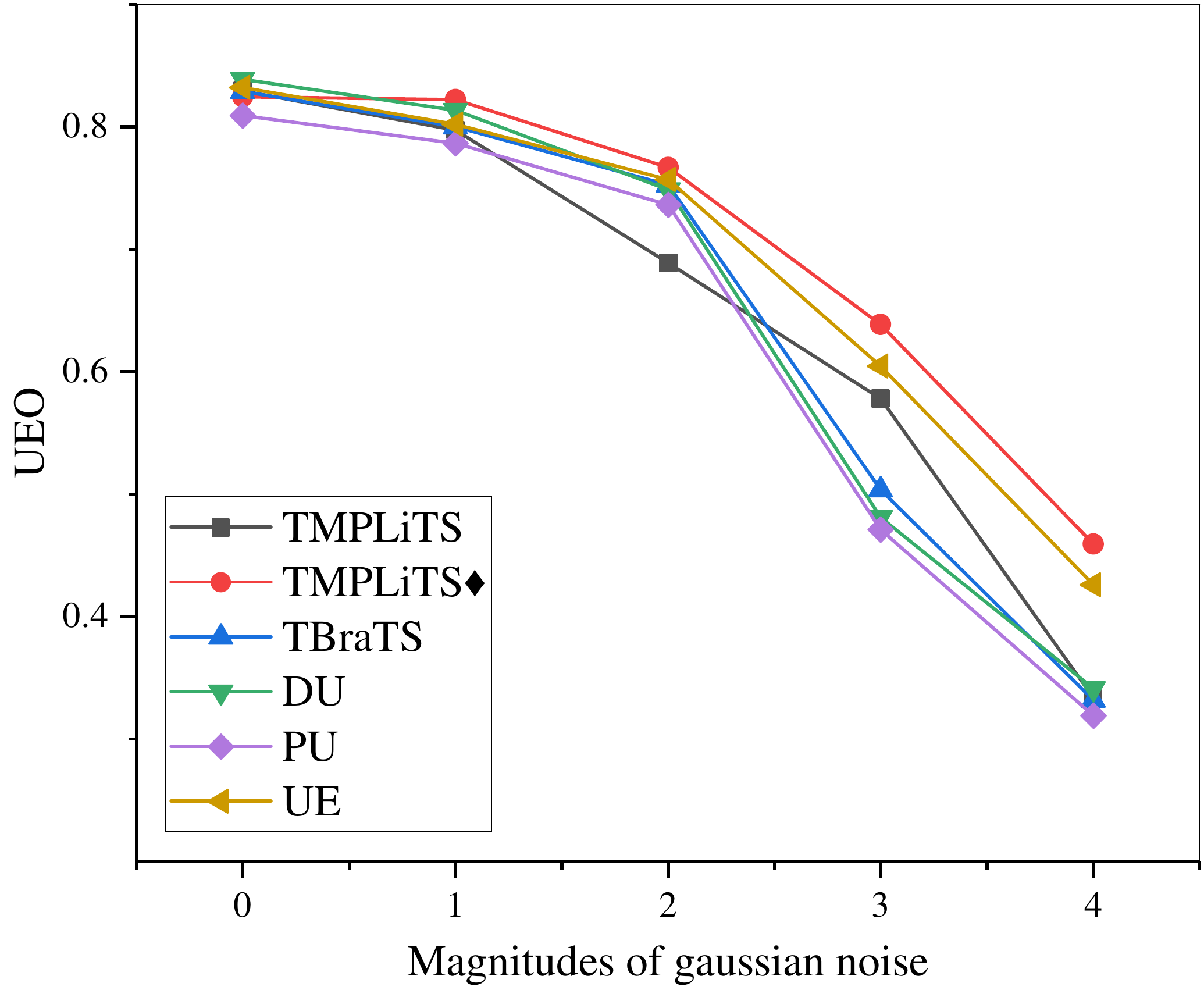}
		\end{minipage}%
		\begin{minipage}[t]{0.48\columnwidth}
			\centering
			%			\hspace{-4.8mm}
			\includegraphics[width=0.98\linewidth]{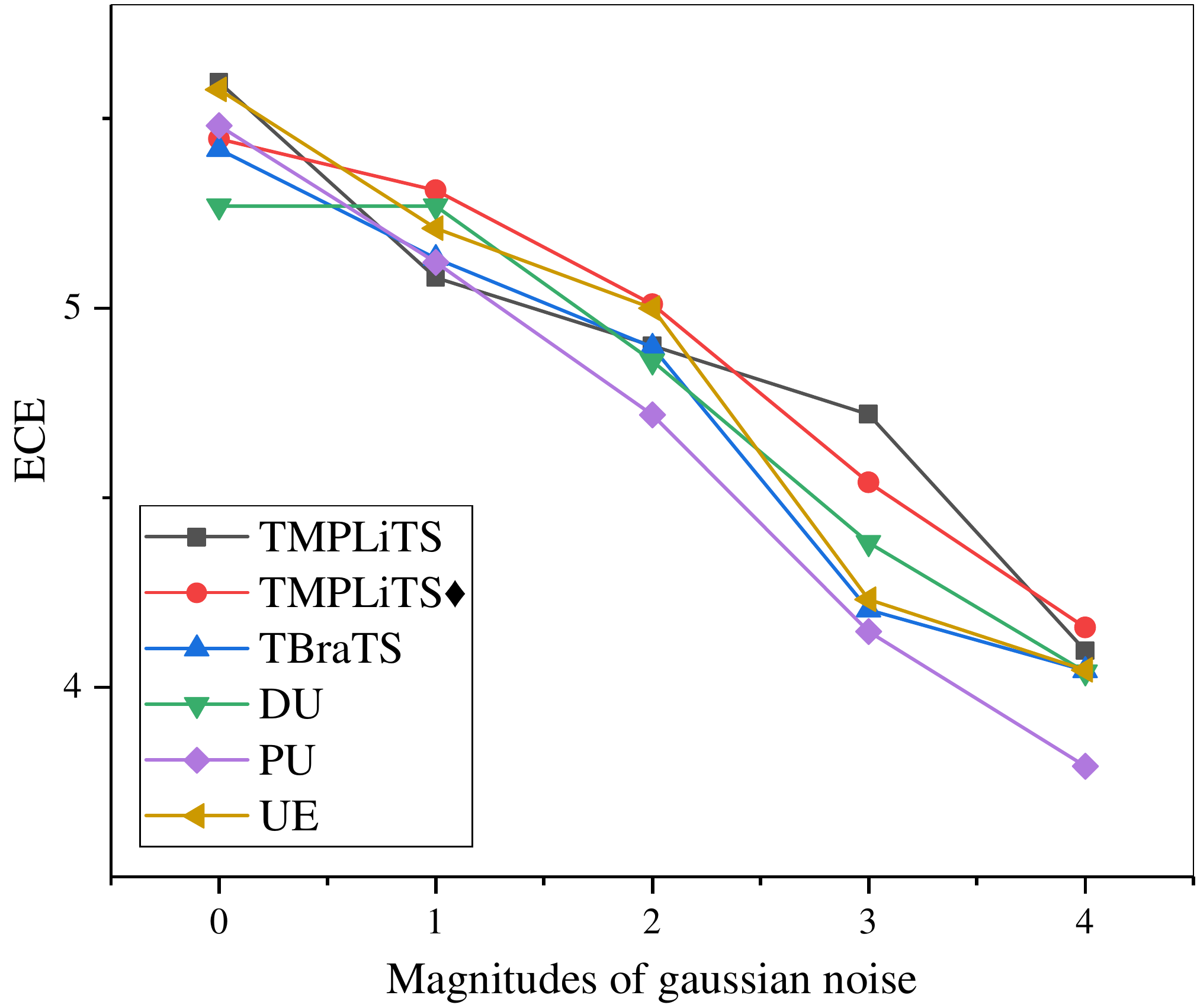}
		\end{minipage}%
	}
	
	\subfigure[Gaussian blur $\mathcal{T}_{\text{\textit{In}}, \text{\textit{Re}}_\text{blur}}$ with $(\sigma^{2}_{b}, k) = \{ (10, 9), (20, 9), (10, 13), (20, 23)\}$. ]{
		\begin{minipage}[t]{0.48\columnwidth}
			\centering
			%			\hspace{-4.8mm}
			\includegraphics[width=1\linewidth]{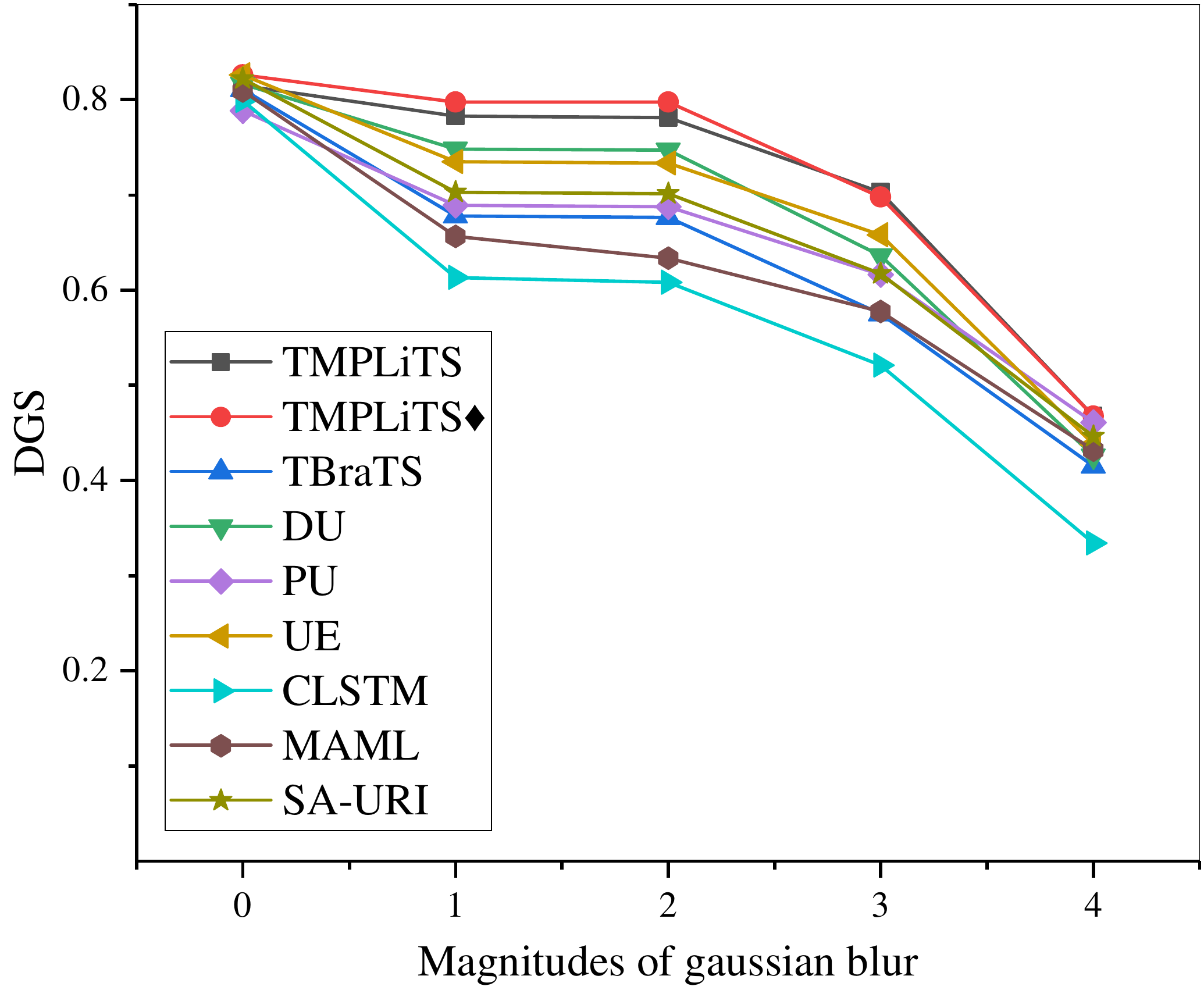}
		\end{minipage}%
		\begin{minipage}[t]{0.48\columnwidth}
			\centering
			%			\hspace{-4.8mm}
			\includegraphics[width=1\linewidth]{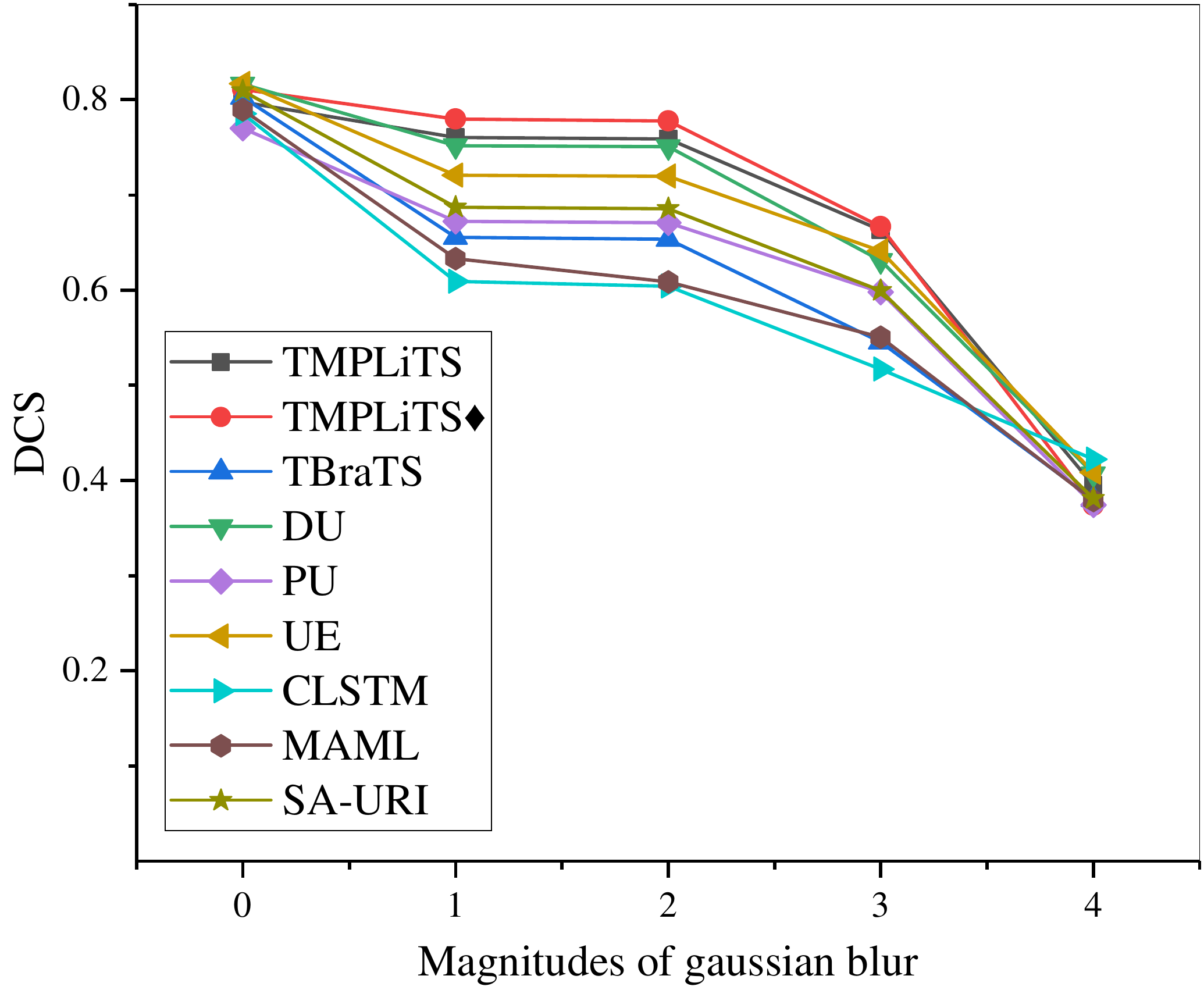}
		\end{minipage}%
		\begin{minipage}[t]{0.48\columnwidth}
			\centering
			%			\hspace{-4.8mm}
			\includegraphics[width=1\linewidth]{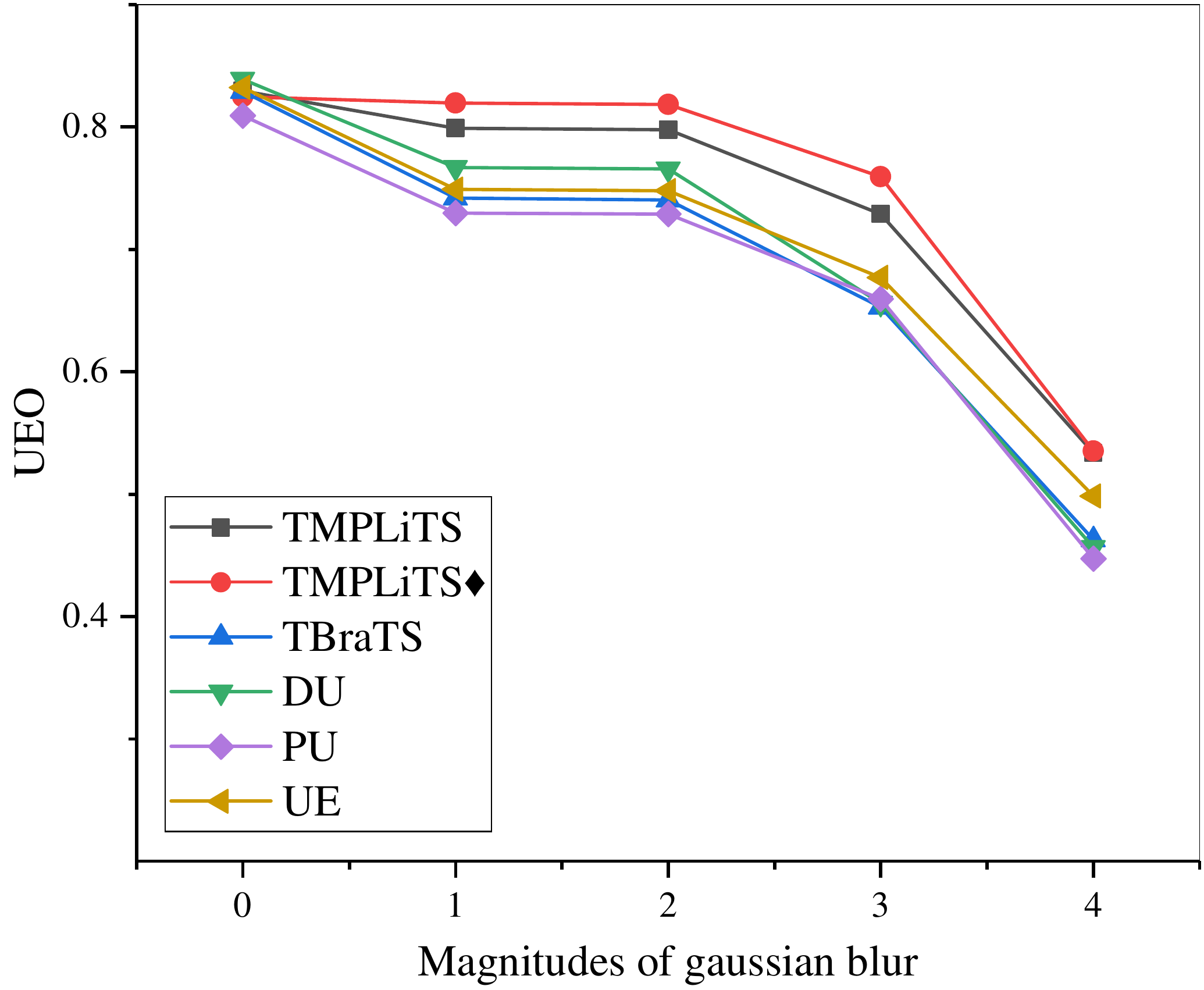}
		\end{minipage}%
		\begin{minipage}[t]{0.48\columnwidth}
			\centering
			%			\hspace{-4.8mm}
			\includegraphics[width=0.98\linewidth]{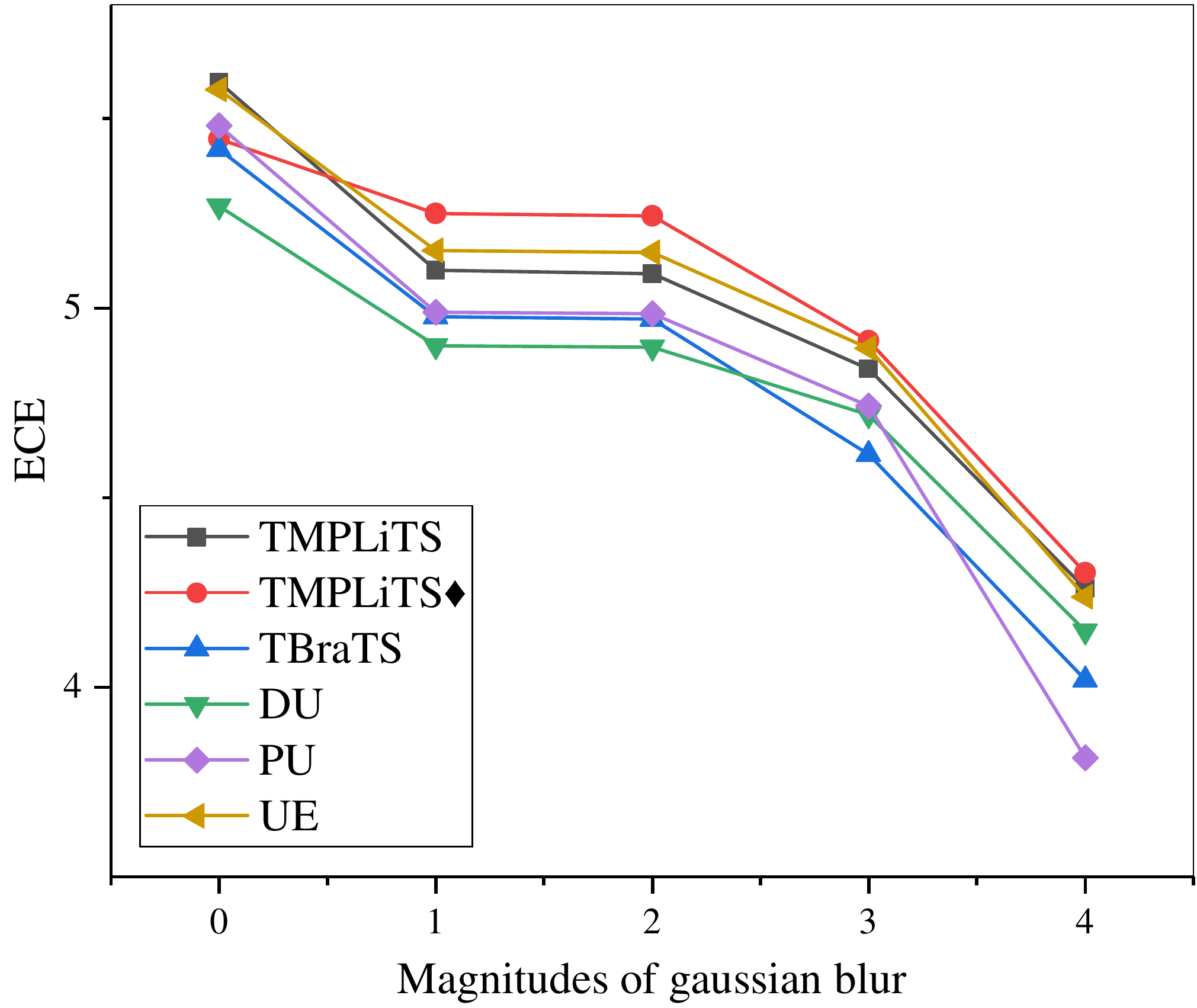}
		\end{minipage}%
	}
	
	\subfigure[Incomplete multi-phase $\mathcal{T}_{\text{\textit{In}}, \text{\textit{Re}}_\text{miss}}$ with $\Omega = \{ 1, 2 \}$. ]{
		\begin{minipage}[t]{0.48\columnwidth}
			\centering
			%			\hspace{-4.8mm}
			\includegraphics[width=1\linewidth]{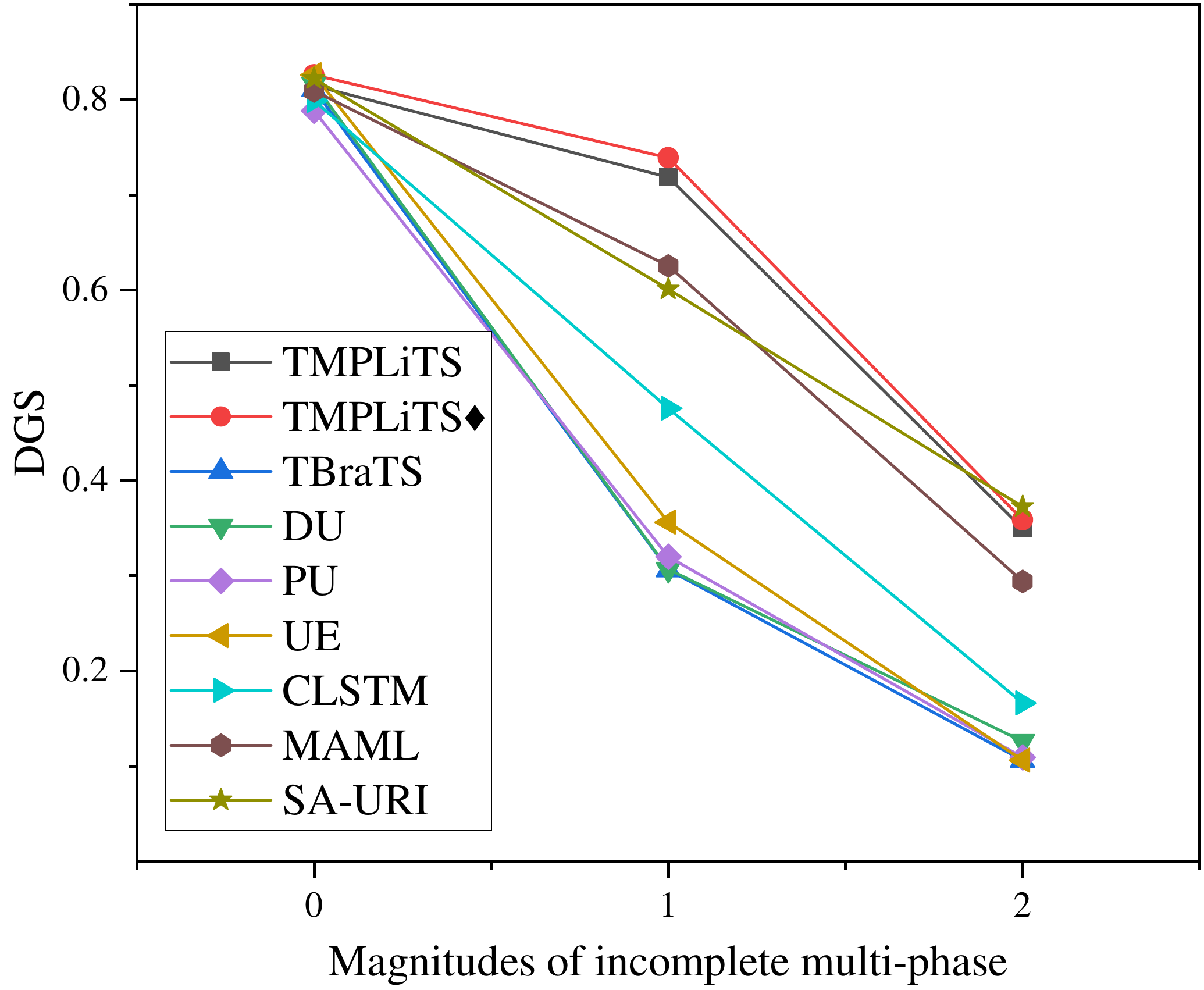}
		\end{minipage}%
		\begin{minipage}[t]{0.48\columnwidth}
			\centering
			%			\hspace{-4.8mm}
			\includegraphics[width=1\linewidth]{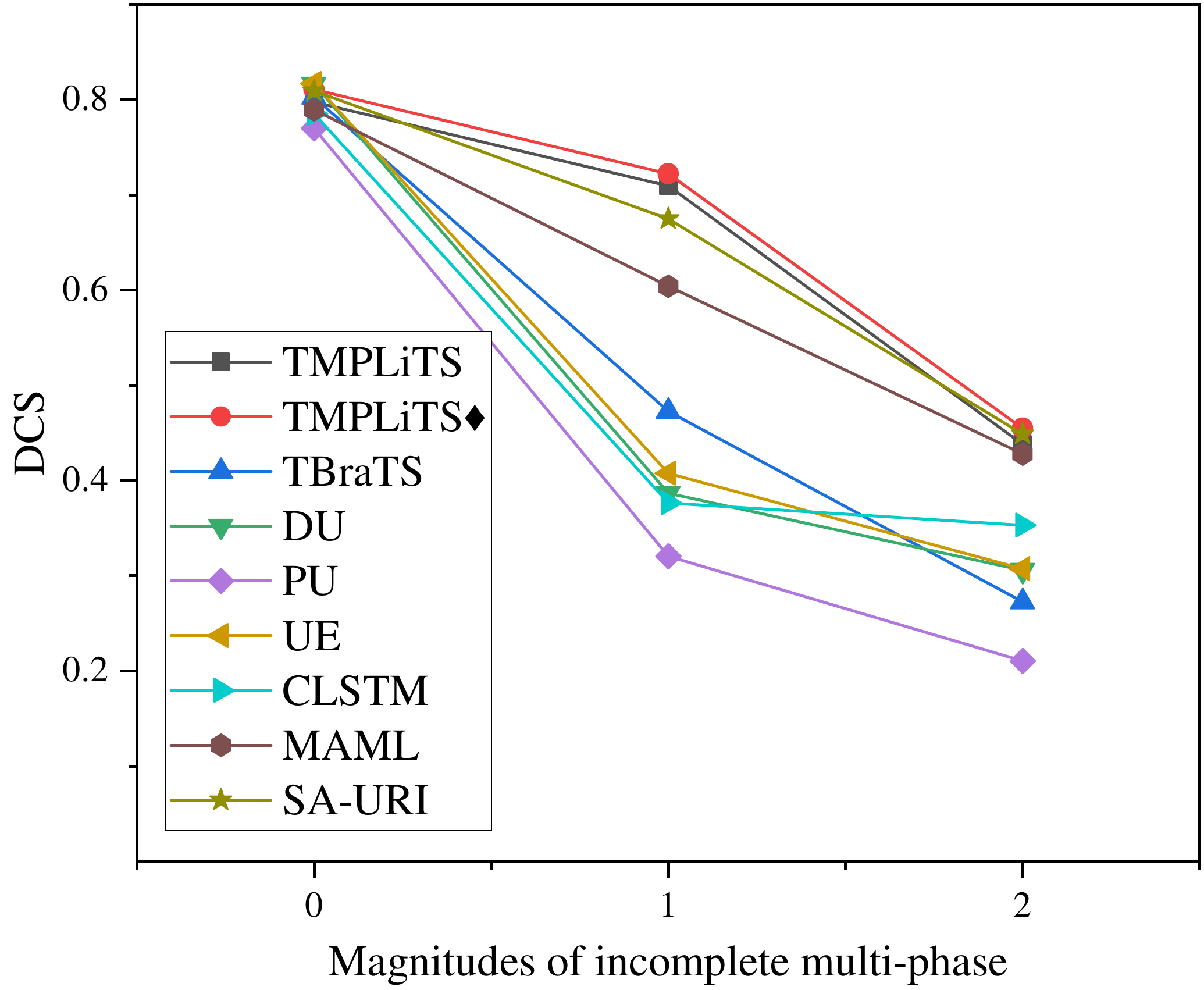}
		\end{minipage}%
		\begin{minipage}[t]{0.48\columnwidth}
			\centering
			%			\hspace{-4.8mm}
			\includegraphics[width=1\linewidth]{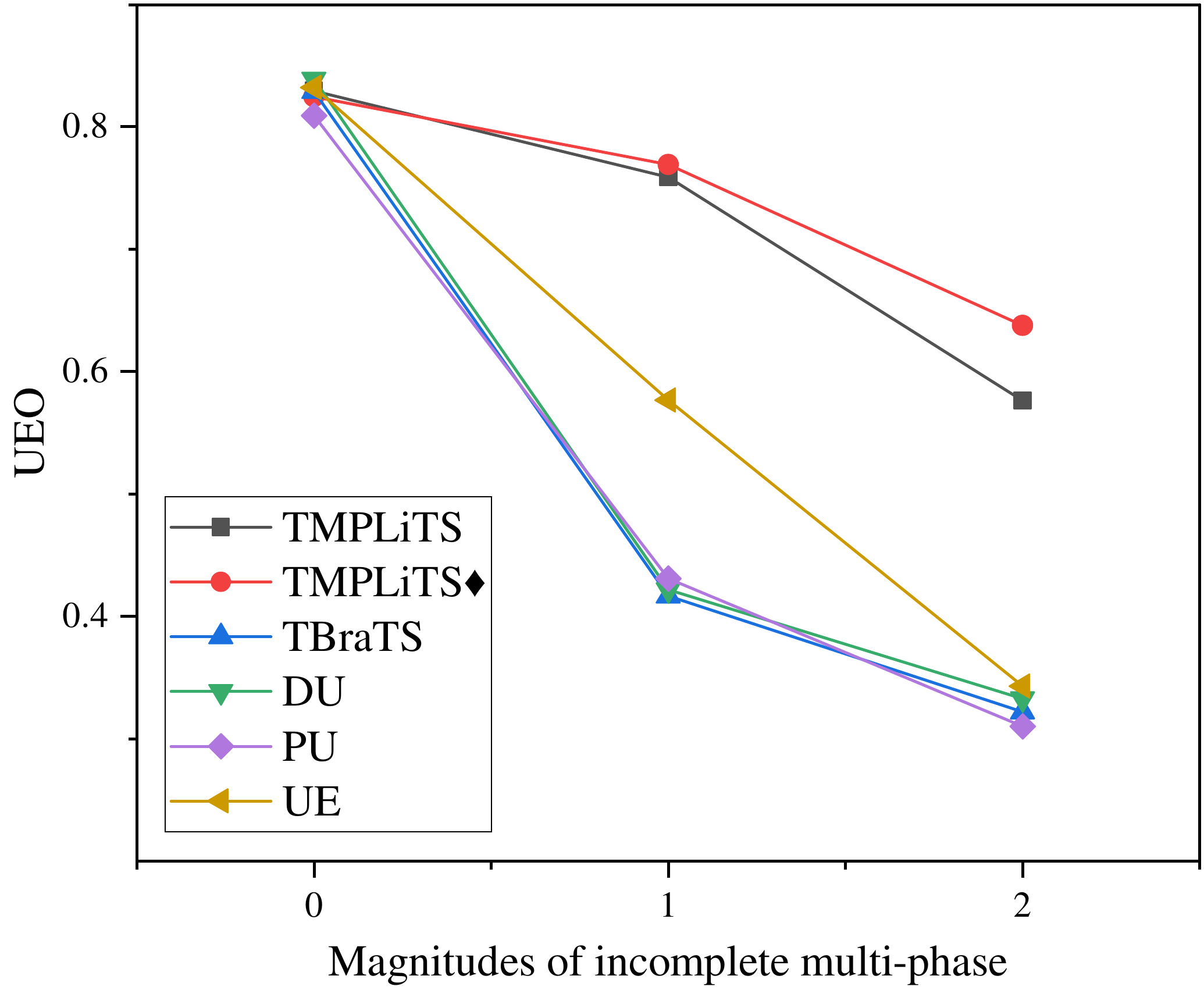}
		\end{minipage}%
		\begin{minipage}[t]{0.48\columnwidth}
			\centering
			%			\hspace{-4.8mm}
			\includegraphics[width=0.98\linewidth]{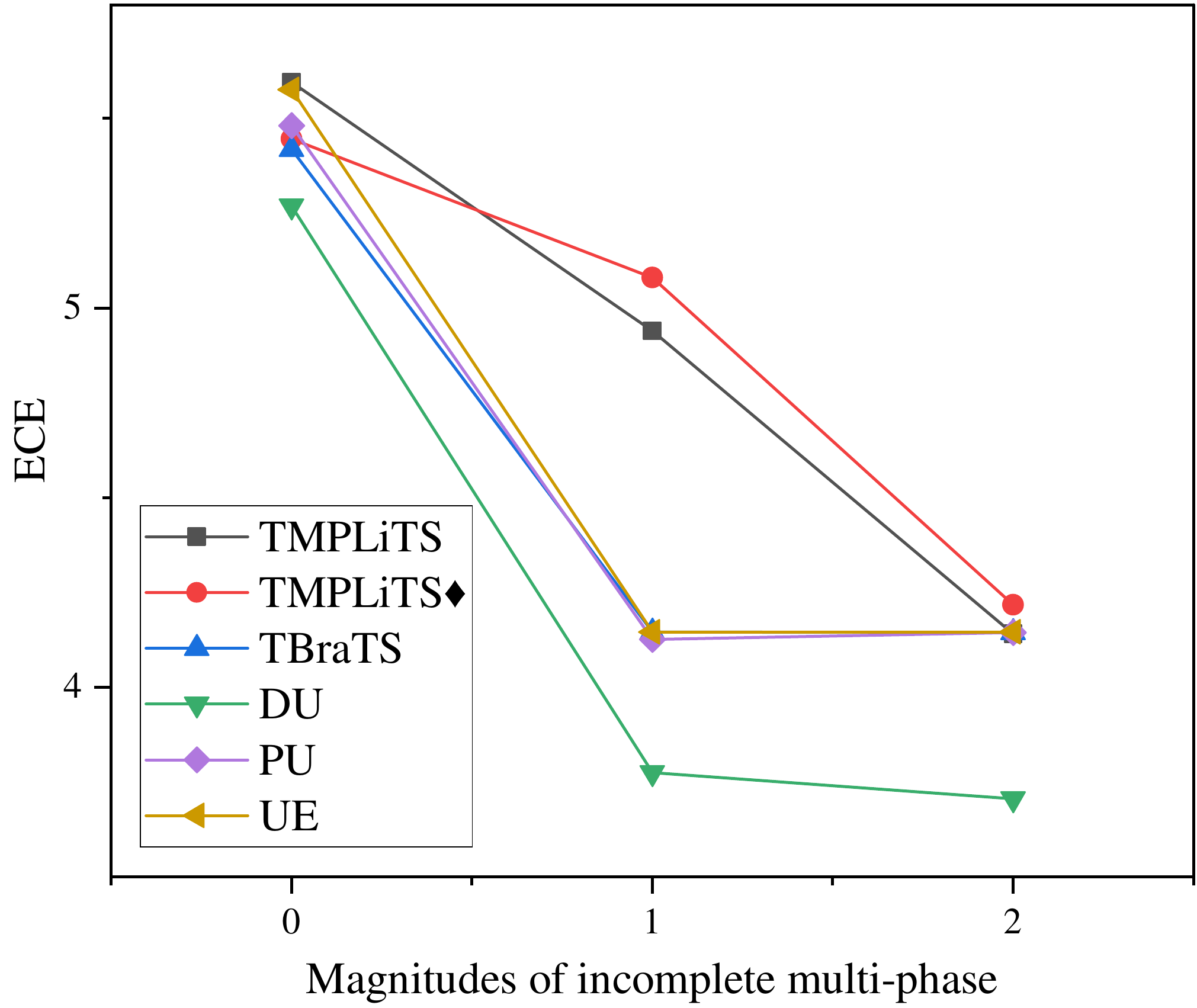}
		\end{minipage}%
	}
	\caption{The quantitative comparisons against the three perturbations with the various magnitudes.} 
	\label{fig:cmp_sota}
	\vspace{-3mm}
\end{figure*}

% tmplits2 > ue > sa_uri > du > tmplits > tbrats > pu > maml > clstm
%tmplits2 > ue > du > tmplits > tbrats > pu 

% $(\sigma^{2}_{b}, k) = (15, 20) $  $(\sigma^{2}_{b}, k) = (15, 20) $ $\Omega = 1 $

%\subsubsection{Evaluation of Reliability on $\mathcal{T}_{\text{\textit{In}}, \text{\textit{Re}}}$ and $\mathcal{T}_{\text{\textit{Ex}}, \text{\textit{Re}}}$} 
\subsubsection{Evaluation of Reliability on $\mathcal{T}_{\text{\textit{In}}, \text{\textit{Re}}}$} 
To verify the reliability of TMPLiTS, we conduct the experiments by simulating the perturbed scenarios. 
Inspired by~\cite{zou2022tbrats}, 
Gaussian noise $\text{\textit{Re}}_\text{noise}$, Gaussian blur $\text{\textit{Re}}_\text{blur}$, and incomplete multi-phase $\text{\textit{Re}}_\text{miss}$ are considered as the perturbed scenarios $\text{\textit{Re}} = \{ \text{\textit{Re}}_\text{noise},  \text{\textit{Re}}_\text{blur}, \text{\textit{Re}}_\text{miss} \}$.
%The generations of these perturbations are described in Algorithm~\ref{app:alg:robust} (in Supplementary Material),
%while the perturbed samples are visualized in Fig.~\ref{xxxx} (in Supplementary Material).
%The generations of these perturbations and the perturbed samples described in Supplementary Material.

To conduct the simulation comprehensively, we control the different magnitudes of perturbations
via the variance coefficients $\sigma^{2}_{n}$, $(\sigma^{2}_{b}, k)$ and incomplete multi-phase $\Omega$ with various constants, respectively.
Here, the variance coefficients $\sigma^{2}_{n} = \{ 0.03, 0.05, 0.1, 0.2 \}$ and $(\sigma^{2}_{b}, k) = \{ (10, 9), (20, 9), (10, 13), (20, 23)\}$ of $\text{\textit{Re}}_\text{noise}$ and $\text{\textit{Re}}_\text{blur}$ are utilized to perturb the multi-phase CECT images. 
% $\sigma^{2}_{n} = \{ 0.05, 0.1, 0.2, 0.3, 0.4\}$
% $(\sigma^{2}_{b}, k) = \{ (11, 10), (13, 10), (15, 20), (23, 20)\}$
The number of missing phases $\Omega$ is set as 1 and 2 to generate the different magnitudes of $\text{\textit{Re}}_\text{miss}$, separately.
%The perturbed samples are visualized in Fig.~\ref{xxx}. % while more perturbed samples  (in Supplementary Material).

We report the performances of trained models against these perturbed scenarios on 
$\mathcal{T}_{\text{\textit{In}}, \text{\textit{Re}}}$.
As shown Fig.~\ref{fig:cmp_sota}, TMPLiTS and its variant present the superior performances than the MPLiTS-based methods in terms of the quantitative metrics.
The performances of some MPLiTS-based methods, 
such as CLSTM and MAML, drop rapidly with the incremental magnitudes of perturbations.
It is noteworthy that SA-URI is more robust than MPLiTS-based methods and some lightweight uncertainty-based methods, such as PU and DU.
The reason is that the capability of complicated model to extract the non-linear representations might be against the slight perturbations based on linear functions.
%while the qualitative results are visualized in Fig.~\ref{fig:cmp_vis}. 
%Meanwhile, 
%the comparisons with other uncertainty-based methods are reported in Fig.~\ref{xxx:b}.
Compared with the other uncertainty-based methods, 
we can observe that the performance of our method is degraded slowly against the perturbations,
while the qualitative results are visualized in Fig.~\ref{fig:cmp_vis}.

These facts verify that the reliability of TMPLiTS and its superiority compared with other uncertainty-based methods.
It is primarily because the procedure of MEMS is based on theoretical guarantees for TMPLiTS.
MEMS reliably fuses the complementary multi-expert opinions from multi-phase to infer the trustworthy results, 
leading to the substantial improvement. 
%More specific results on $\mathcal{T}_{\text{\textit{In}}, \text{\textit{Re}}}$ and $\mathcal{T}_{\text{\textit{Ex}}, \text{\textit{Re}}}$ in terms of qualitative and quantitative comparisons are reported in Supplementary Material.
%designed via evidence-based uncertainty for
%It is noteworthy that the comparison results on $\mathcal{T}_{\text{\textit{In}}, \text{\textit{Re}}}$ also argue the above observation.
%on $\mathcal{T}_{\text{\textit{Ex}}, \text{\textit{Re}}}$, reported in Fig.~\ref{xxx} and Table~\ref{xxx} (in Supplementary Material), also argue the above observation.

% table1: comparison of all methods on in-house and external dataset in terms of DGC and DCS.
% Fig.1 framework of TMPLiTS.
% Fig.2 line plots with error dai, 4 metrics * 3 perturbations for all methods = array 3 rows * 4 cols.
% Fig.3 robust images compared with MPLiTS-based methods, (1 four cat (org_imgs) + 4 methods + 1 gt) * 4 scenarios = array 4 rows * 6 cols
% Appendix:
% Fig.a1 the figure of T_ex,va 
% Table a1 the table of T_ex,va 
% Algorithm 1 of three perturbations 

% we also add different levels of Gaussian noise, Gaussian blur and random mask

\begin{figure*}[!t]
	\centering
	\includegraphics[width=\textwidth]{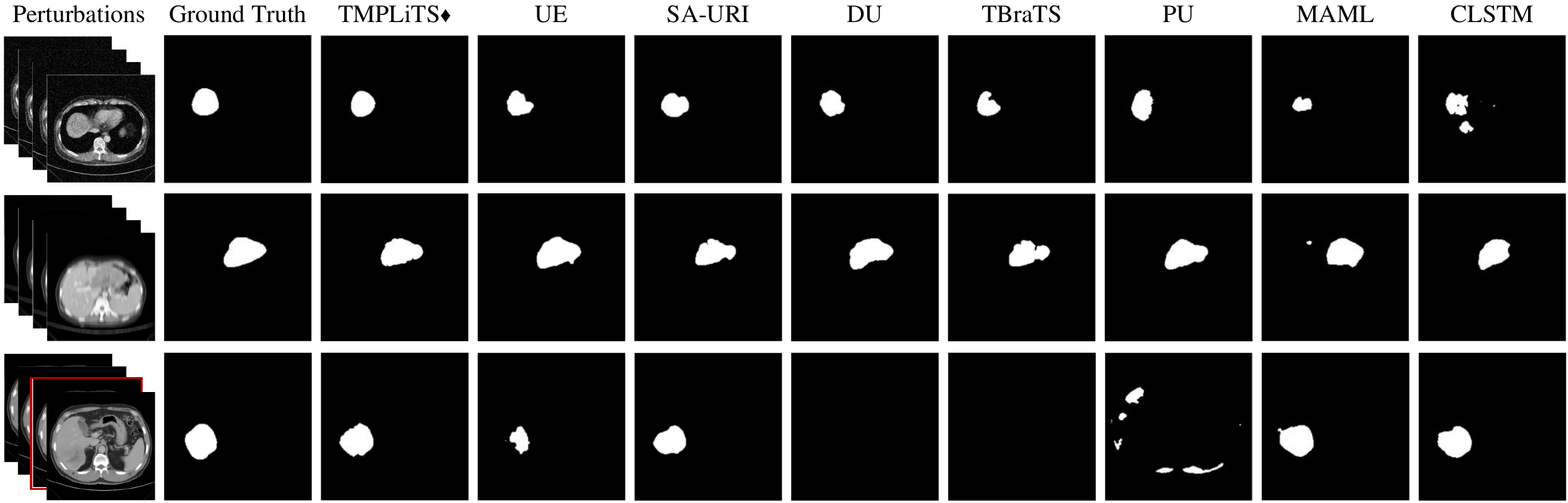}
	\caption{The qualitative comparisons against the three perturbations.
	The first column illustrates the three examples of the difference perturbations, 
	including $\text{\textit{Re}}_\text{noise}$, $\text{\textit{Re}}_\text{blur}$, and $\text{\textit{Re}}_\text{miss}$.
	The variance coefficients of these perturbations are $\sigma^{2}_{n} = 0.1$, $(\sigma^{2}_{b}, k) = (10, 13)$, and $\Omega=1$, respectively.
	The results illustrate the reliability of TMPLiTS qualitatively.
} 
	\label{fig:cmp_vis}
	\vspace{-4mm}
\end{figure*}

\begin{table*}[!t]
%	\definecolor{whitesmoke}{gray}{.95}
	\definecolor{whitesmoke}{RGB}{239,246,255}
	\centering
	\caption{Ablation studies of TMPLiTS. %on $\mathcal{T}_{\text{\textit{Ex}}, \text{\textit{Va}}}$ and $\mathcal{T}_{\text{\textit{Ex}}, \text{\textit{Re}}}$. 
		The average results and standard deviations are both reported for five-fold cross-validation.
		``w/ $F^{s}(\cdot)$'' denotes TMPLiTS with independent feature extractors.
		``w/o MEMS'' presents TMPLiTS without MEMS, while an average operation is utilized to replace the pixel-wise reduced Dempster's combination.
	}\label{tab:ablation}
%	\begin{tabular}{@{}cc|ccc|ccc|ccc@{}}
	\begin{tabular}{cc|ccc|ccc|ccc}
%	\begin{tabular}{ccccccccccc}
		\Xhline{1pt} & & & & & & & & & & \\[-3mm] 
		\multicolumn{2}{c|}{\multirow{2}{*}{Methods}} & \multicolumn{3}{c|}{w/ $F^{s}(\cdot)$} & \multicolumn{3}{c|}{ w/o MEMS} & \multicolumn{3}{c}{ TMPLiTS} \\
		\multicolumn{2}{c|}{}                          & DGS  & UEO &  ECE  & DGS  & UEO  & ECE & DGS & UEO & ECE \\
%		\rowcolor{whitesmoke}
		\hline %& & & & & & & & & & \\[-2.5mm]
		\rowcolor{whitesmoke}
		\multicolumn{2}{c|}{$\mathcal{T}_{\text{\textit{In}}, \text{\textit{Va}}}$}                & \textbf{82.60}(1.68)   & 82.45(1.69)    & 5.45(0.13)   & 80.39(0.70)   & 81.46(0.75)    & 5.44(0.12)  & 81.49(1.74)    & \textbf{82.95}(1.58)   & \textbf{5.60}(0.07) \\
		\multicolumn{2}{c|}{$\mathcal{T}_{\text{\textit{Ex}}, \text{\textit{Va}}}$}                & \textbf{79.20}(0.86)    & \textbf{80.10}(0.53)   & \textbf{5.63}(0.06)   & 77.44(1.02)   & 78.13(1.01)    & 5.50(0.11)  & 78.18(0.82)   & 78.78(0.93)    & 5.54(0.06)\\
%		\rowcolor{whitesmoke}
		\hline %& & & & & & & &&&\\[-2.5mm]
		\rowcolor{whitesmoke}
		$\mathcal{T}_{\text{\textit{In}}, \text{\textit{Re}}_\text{noise}}$              & $0.03 $               & \textbf{82.12}(1.86) & \textbf{82.22}(1.83) &\textbf{5.31}(0.11) & 72.98(1.48) & 72.87(1.49) & 4.91(0.19) & 79.19(1.26) & 79.75(1.41) & 5.08(0.11)   \\
		$\mathcal{T}_{\text{\textit{In}}, \text{\textit{Re}}_\text{noise}}$              & $0.05 $               & \textbf{73.90}(2.21) & \textbf{76.69}(2.00) &\textbf{5.01}(0.24) & 69.80(1.28) & 69.78(2.35) & 4.84(0.17) & 73.00(2.08) & 68.90(2.31) & 4.90(0.08)   \\
		\rowcolor{whitesmoke}
		$\mathcal{T}_{\text{\textit{In}}, \text{\textit{Re}}_\text{noise}}$              & $0.10 $               & \textbf{57.00}(2.77) & \textbf{63.84}(3.41) & 4.54(0.17) &  44.80(1.52) & 60.43(2.40) &4.42(0.24) & 53.19(6.16) & 57.26(2.73) & \textbf{4.72}(0.14)   \\
		%		$\mathcal{T}_{\text{\textit{In}}, \text{\textit{Re}}_\text{noise}}$              & $0.20 $               & 1    & 2   & 3   & 4   & 5    & 6  & 4   & 5    & 6   \\
		%		\rowcolor{whitesmoke}
		\hline %& & & & & & & &&&\\[-2.5mm]
		$\mathcal{T}_{\text{\textit{In}}, \text{\textit{Re}}_\text{blur}}$              & $(10, 9) $               &\textbf{79.93}(1.35) &\textbf{81.95}(2.13) &\textbf{5.24}(0.20) &76.09(1.94) &76.00(2.96) &5.06(0.05) &78.30(1.68) &79.96(1.08) &5.10(0.17)   \\
		\rowcolor{whitesmoke}
		$\mathcal{T}_{\text{\textit{In}}, \text{\textit{Re}}_\text{blur}}$              & $(20, 9) $               &\textbf{79.76}(1.37) &\textbf{81.85}(2.14) &\textbf{5.24}(0.20) &75.95(1.93) &75.86(1.96) &5.06(0.05) &78.16(1.72) &79.83(1.10) &5.09(0.17)   \\
		$\mathcal{T}_{\text{\textit{In}}, \text{\textit{Re}}_\text{blur}}$              & $(10, 13) $              &69.78(3.88) &\textbf{75.94}(1.98) &\textbf{4.91}(0.23) &67.66(3.21) &67.78(3.24) &4.58(0.08) &\textbf{70.31}(3.49) &73.24(2.26) &4.84(0.21)   \\
		\rowcolor{whitesmoke}
		$\mathcal{T}_{\text{\textit{In}}, \text{\textit{Re}}_\text{blur}}$              & $(20, 23) $               &\textbf{46.77}(3.66) &\textbf{54.02}(2.71) &\textbf{4.32}(0.17) &38.95(2.84) &43.51(3.77) &4.10(0.21) &44.69(3.35) &46.18(3.28) &4.26(0.07)   \\
		\hline %& & & & & & & &&&\\[-2.5mm]
		$\mathcal{T}_{\text{\textit{In}}, \text{\textit{Re}}_\text{miss}}$              & $1 $               &\textbf{73.89}(1.94) &\textbf{76.91}(1.05) &\textbf{5.08}(0.31) &70.71(1.12) &74.25(1.23) &4.85(0.20) &71.89(1.97) &75.91(1.93)  &4.94(0.19)   \\
		\rowcolor{whitesmoke}
		\hline %& & & & & & & &&&\\[-2.5mm]
		\rowcolor{whitesmoke}
		\multicolumn{2}{c|}{Memory footprint} &\multicolumn{3}{c|}{ 205.3 MB} &\multicolumn{3}{c|}{ \textbf{51.7} MB} &\multicolumn{3}{c}{\textbf{51.7} MB}   \\
		%		\rowcolor{whitesmoke}
		%		$\mathcal{T}_{\text{\textit{In}}, \text{\textit{Re}}_\text{miss}}$              & $2 $               & 1    & 2   & 3   & 4   & 5    & 6    & 4   & 5    & 6 \\
		%		$\mathcal{T}_{\text{\textit{Ex}}, \text{\textit{Re}}_\text{miss}}$              & $3 $               & 1    & 2   & 3   & 4   & 5    & 6    & 4   & 5    & 6 \\
		\Xhline{1pt}
	\end{tabular}
\vspace{-3mm}
\end{table*}

%\begin{table}[]
%	\begin{tabular}{@{}cccc}
%		&           &  &  \\
%		[0mm] \hline & & & \\[-2.5mm]
%		\rowcolor[HTML]{EFEFEF} 
%		\multicolumn{2}{c}{1} &  &  \\
%		[0mm] \hline & & & \\[-2.5mm]
%		\rowcolor[HTML]{EFEFEF} 
%		$\mathcal{T}_{\text{\textit{In}}, \text{\textit{Re}}_\text{noise}}$         & 4         &  &  \\
%		$\mathcal{T}_{\text{\textit{In}}, \text{\textit{Re}}_\text{noise}}$&           &  & 
%	\end{tabular}
%\end{table}

\subsection{Ablation Studies}

% share feature extractor or ?
\subsubsection{Analysis of Shared Feature Extractor}
To investigate the characteristic of shared feature extractor $F(\cdot)$, 
we model a variant of TMPLiTS (TMPLiTS$\blacklozenge$) with the independent feature extractors.
%The $s$-th feature extractor $F^{s}(\cdot)$ is conducted independently for the $s$-th phase.
As reported in Table~\ref{tab:ablation}, the architecture of independent feature extractors improves the expected performances in terms of validity and reliability.
It means that the discriminative representations of each phase are captured by the corresponding feature extractors. 
However, the additional parameters of model are inevitable intuitively, 
where the memory footprints of TMPLiTS with $F(\cdot)$ and $F^{s}(\cdot)$ are 51.7 MB and 205.3 MB, respectively. 
Thus, the architecture of shared feature extractor $F(\cdot)$ can be seen as a trade-off between accuracy and complexity.

\subsubsection{Effect of Multi-Expert Mixture Scheme}
To clarify the effect of MEMS, a variant of TMPLiTS without MEMS is designed, 
where the pixel-wise reduced Dempster's combination is replaced by average operation to obtain the joint opinion $\mathcal{M}_{i,j}$ from the multi-phase opinions.
As reported in Table~\ref{tab:ablation}, the performance of TMPLiTS without MEMS is decreased obviously compared with TMPLiTS.
The pixel-wise DST-based combination rule outperforms the average operation about 10\% and 6\% against $\text{\textit{Re}}_\text{noise}$ and $\text{\textit{Re}}_\text{blur}$ in terms of DGS, respectively.
%There are approximately.
It reveals the effect of MEMS to achieve the reliable fusion procedure for multi-phase opinions.
%It can be seen that the tendency of the performances with incremental

\begin{figure}[!t]
	\centering
	\includegraphics[width=0.90\columnwidth]{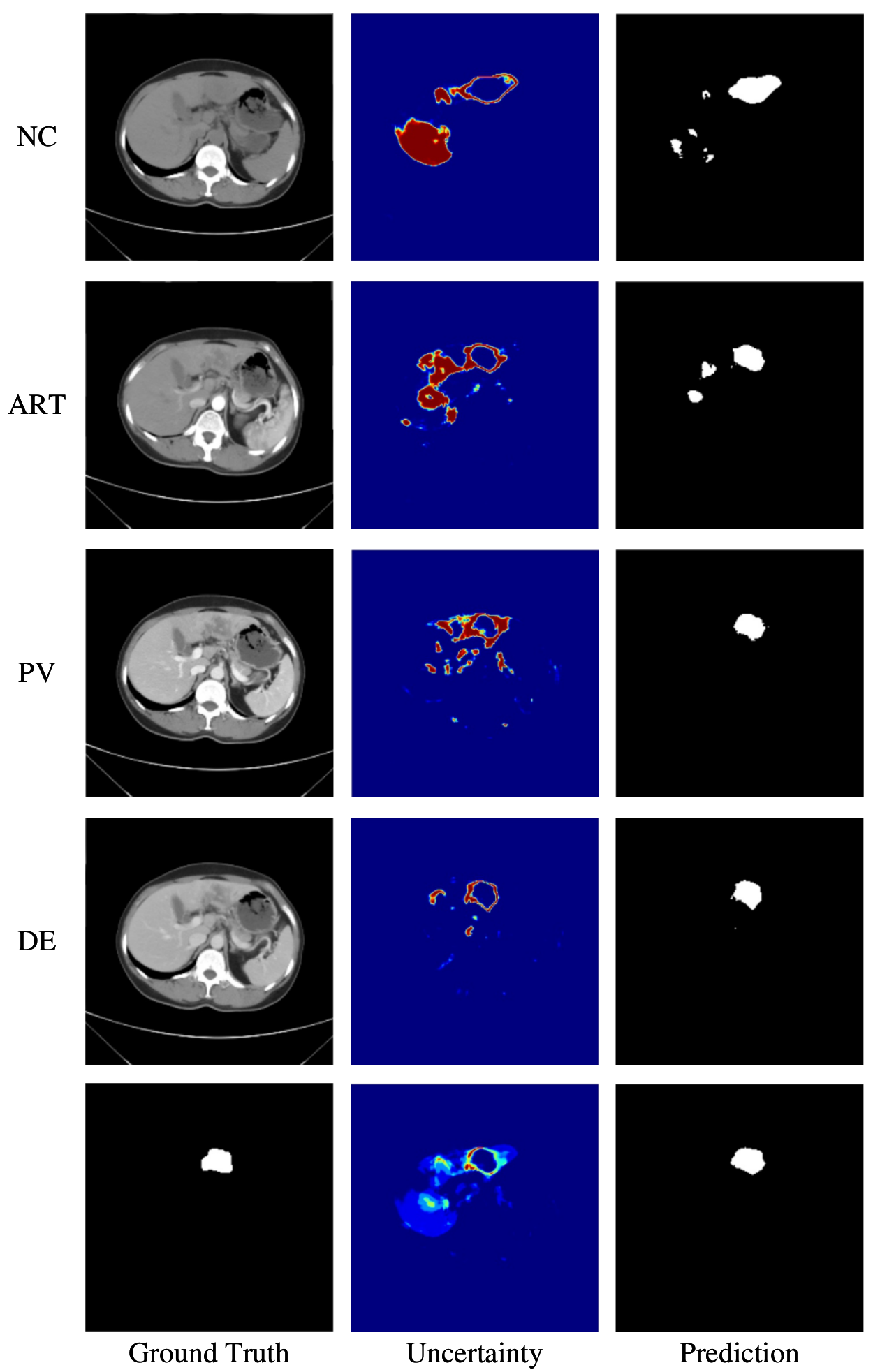}
	\caption{Visualization of trustworthy results generated by TMPLiTS. 
		The first four rows illustrate the multi-phase results, 
		including the CECT images, heatmaps of uncertainty, and predictions.
		The ground truth and fused results are visualized in the last row.
}
	\label{fig:interpretation}
	\vspace{-3mm}
\end{figure}

\begin{figure}[!t]
	\centering
	%	\vspace{-4mm}
	%	\hspace{-2mm}
	\subfigure[Correlation on $\mathcal{T}_{\text{\textit{In}}, \text{\textit{Va}}}$, where correlation coefficient is 0.969.]{
		\begin{minipage}[t]{1\columnwidth}
			\centering
			\hspace{-4.8mm}
			\includegraphics[width=0.7\linewidth]{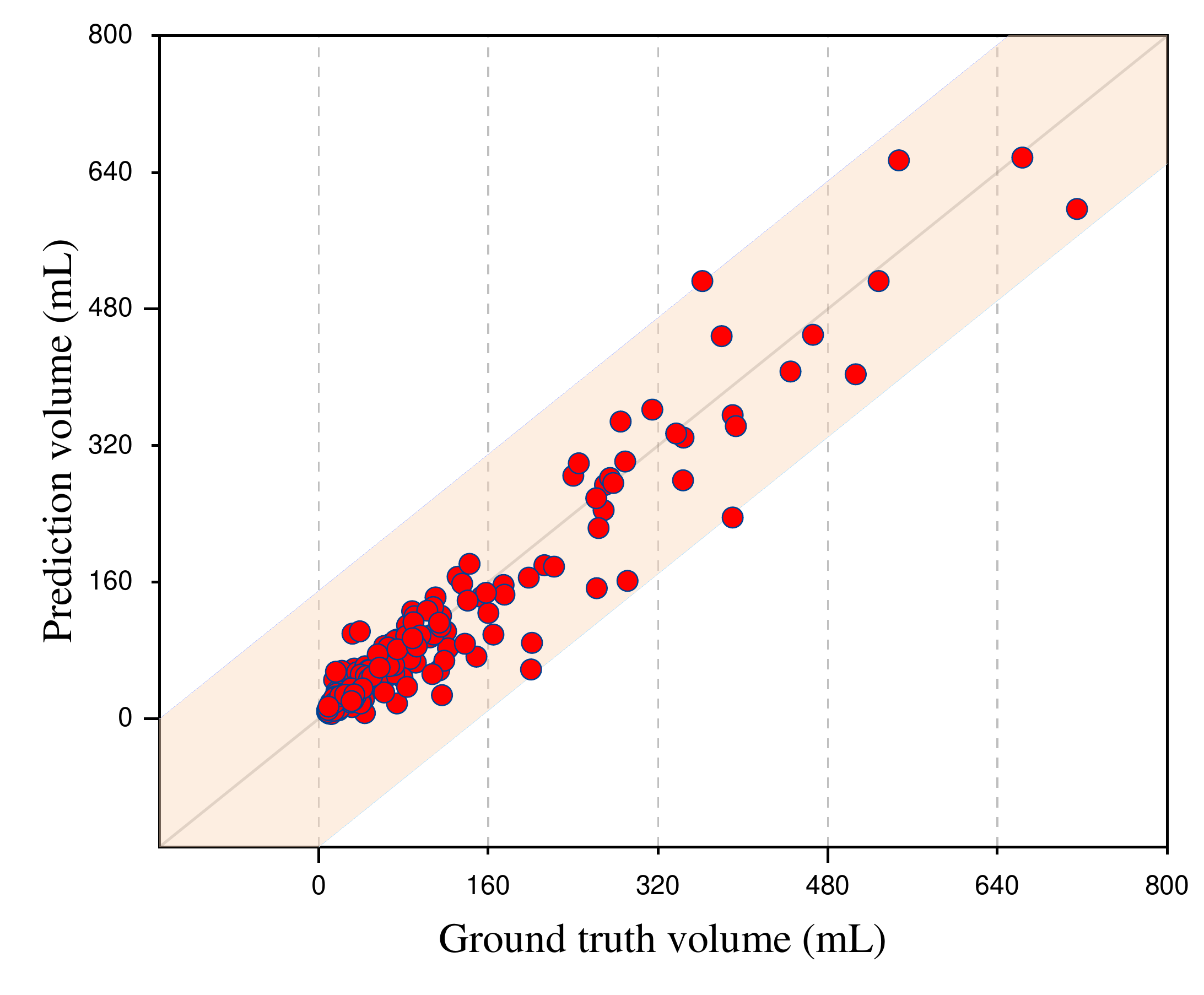}
		\end{minipage}%
	}
	\hspace{2mm}
	\subfigure[Correlation on $\mathcal{T}_{\text{\textit{Ex}}, \text{\textit{Va}}}$, where correlation coefficient is 0.961.]{
		\begin{minipage}[t]{1\columnwidth}
			\centering
			\hspace{-4.8mm}
			\includegraphics[width=0.7\linewidth]{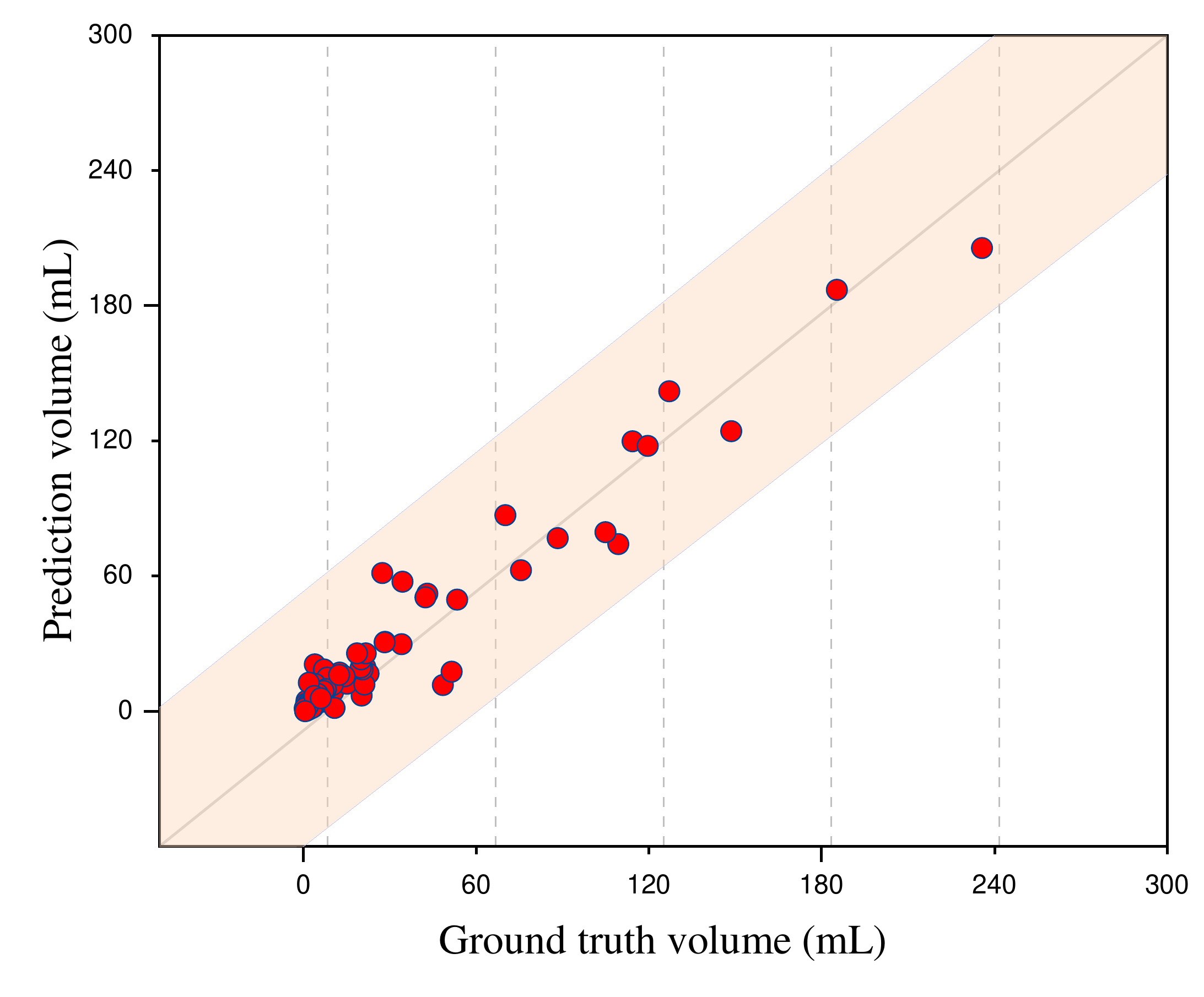}
		\end{minipage}%
	}
	\caption{The correlation plots between TMPLiTS and clinicians for the tumor volumes of each patient.}
	\label{fig:dis_correlation}
	\vspace{-2mm}
\end{figure}

%\subsubsection{Interpretation of each phase}
%Furthermore, we evaluate the performances 
\subsection{Discussion}

\subsubsection{Interpretation of multi-phase trustworthy results}
Clinically, during the hepatic arterial phase, lesions would be greatly enhanced, 
and become iso- or hypodense in the portal venous phase, which is a sensitive and specific characteristic for diagnosing HCC~\cite{mathieu1985portal}.
The delayed phase is helpful in the representation of hepatocellular carcinoma by depiction of a capsule or mosaic pattern~\cite{lim2002detection}.
In order to interpret the multi-phase trustworthy results of TMPLiTS qualitatively, 
we visualize the uncertainty and the corresponding predictions, as shown in Fig.~\ref{fig:interpretation}. 
The high values of heatmap regions represent the high magnitudes of uncertainty, 
which could suggest the clinician to pay more attention to these regions of each phase.
In the non-contrast phase, the initial regions of lesions are predicted.
The suspicious regions are highlighted by the heatmap, which should be focused carefully.
With the enhancement of hepatic parenchyma during the other phases, 
the regions of lesions have shrunk to the precise regions. %, while a capsule of tumor in the delayed phase.
Finally, the heatmap of the fused uncertainty describes the reliability of prediction, 
such as incomplete capsule and suspicious patterns.
These observations reveal the instructive assistance of trustworthy results provided by TMPLiTS.
%In conclusion, delayed phase CT is useful in the detection and characterization of hepatocellular carcinoma, especially for the detection of small hepatocellular carcinomas

% strengthens 
%\subsubsection{Importance of multi-phase CECT images}
\subsubsection{Analysis of correlation}

%The left column is the correlation plot 
%which shows the correlation between the predicted volume and the  true volume of LV (LVV) and RV (RVV) in ED and ES.
%The correlation coefficients of LVV and RVV are 0.986 and 0.989.

To assess the overall agreement between TMPLiTS and clinicians on 
$\mathcal{T}_{\text{\textit{In}}, \text{\textit{Va}}}$ and $\mathcal{T}_{\text{\textit{Ex}}, \text{\textit{Va}}}$, the analysis of correlation is introduced. 
The tumor volumes of each patient case are converted from pixel-wise results based on spacing.
As shown in Fig.~\ref{fig:dis_correlation}, the correlation between the prediction and ground truth in terms of volumes are plotted,
where the correlation coefficients on $\mathcal{T}_{\text{\textit{In}}, \text{\textit{Va}}}$ and $\mathcal{T}_{\text{\textit{Ex}}, \text{\textit{Va}}}$ are 0.969 and 0.961, respectively.
It can be observed that the correlation points of prediction and ground truth volumes are closed among the diagonal compactly.
Meanwhile, the performance of TMPLiTS on the various magnitudes of tumor volumes is acceptable.
These facts indicate that the predicted results obtained by the proposed method are in high agreement with the real results.
%The right column is the Bland-Altman graph which shows the different distributions between the predicted and
%true volumes along the means between the predicted and true volumes. 
%The means bias values are -0.14ml (with confidence intervals between 5.74 and 1.96), 
%-2.21ml (with confidence intervals between 4.65 and 9.07) in terms of LVV and RVV.
 
\subsubsection{Reliability of model}

TMPLiTS guarantees the reliability of model by the theoretical supports.
However, the model with high complexity, such as SA-URI, is also against the slight perturbations, as shown in Fig.~\ref{fig:cmp_sota}.
The insight behind the results is that such slight perturbations based on linear functions could be tolerant by the complicated feature extraction.
Compared with the high complicated model, 
the advantage of TMPLiTS is to achieve the reliable liver tumor segmentation and uncertainty estimation efficiently.
%intentionally
 
%The reason is that the capability of complicated model to extract the non-linear representations might be against the slight perturbations based on linear functions.
 
%\section{Conclusion}
\subsubsection{Limitation}

Although TMPLiTS has achieved a satisfactory performance, 
there are some limitations which could be further improved. 
First, the high magnitudes of perturbations need to be further considered in the procedure of modeling. 
For instance, the collection of the multi-phase CECT volumes with two or more incomplete phases is a common case, 
while the existing methods might not present the reliable performance towards the setting of missing more incomplete phases.
Second, the mechanism of TMPLiTS behind the balance between accuracy and complexity could be analyzed,
which might enrich the properties of TMPLiTS.
Finally, the generalization in the other medical images could be further explored,
since the similar insight of modeling the multi-phase complementarity might be introduced by TMPLiTS.

%\textred{Need to adjust the following limitations,
%including high magnitude of perturbations, the balance between accuracy and complexity, generalizability for other medical images.}

%First, more multi-center external test sets need to be collected to further verify the generalizability of the proposed model. 
%Second, co-registration errors may persist between multi-phase DCE images due to motion artifacts in some cases, 
%which may interfere with segmentation models based on co-registered multi-phase images. 
%Third, since the size of liver tumors is variable, 
%whole liver region based tumor segmentation strategy could potentially miss small tumors, 
%and patch-wise segmentation can be considered for this case. 
%Finally, our feature map analysis reflects some useful information, 
%but some of the conclusions are based on human interpretation and speculation with reference to the experience of radiologists. 
%More quantitative experiments are to be combined in future work for a more detailed feature map analysis.
%\vspace{1cm}
\section{Conclusion}
In this paper, we propose a novel unified framework via evidence-based uncertainty for MPLiTS, 
termed as trustworthy multi-phase liver tumor segmentation (TMPLiTS), 
where the reliable segmentation and uncertainty estimation are jointly conducted on the multi-phase CECT images.
The complementary multi-phase information are fused based on multi-expert mixture scheme (MEMS) with theoretical guarantees, 
while the interpretable liver tumor segmentations of each phase are obtained independently,  
assisting the diagnosis of liver cancer clinically.
Experimental results demonstrate that TMPLiTS achieves the competitive performance compared with the state-of-the-art methods. 
Moreover, the robustness of TMPLiTS is verified against the three perturbed scenarios.

In future work, we will extend the insight of TMPLiTS to the other medical images,
while the architecture and mechanism of TMPLiTS will be further explored to alleviate the high magnitudes of perturbations, such as two or more incomplete phases.

\vspace{-2mm}
\bibliographystyle{IEEEtranTIE}
\bibliography{20221222_tmi}

\end{document}